\documentclass{acm_proc_article-sp}
\newcommand{\remove}[1]{}
\usepackage{multicol}
\usepackage{tabulary}
\usepackage[
    font={small},
            format=hang,    
                format=plain,
                    margin=0pt,
                        width=1.05\textwidth,labelfont=bf,tableposition=top
                        ]{caption}
\usepackage[list=true,font={small}]{subcaption}
\usepackage{algorithm}
\usepackage{algpseudocode}
\makeatletter
\renewcommand{\ALG@beginalgorithmic}{\small}
\makeatother

\algnotext{EndFor}
\algnotext{EndIf}
\algrenewcomment[1]{\(\triangleright\) #1}
\algnewcommand{\LineComment}[1]{\State \(\triangleright\) #1}

\usepackage{graphicx}
\usepackage[T1]{fontenc}
\usepackage{scalefnt}
\usepackage{float}
\usepackage{amssymb}

\usepackage{amsthm}
\usepackage{epstopdf}
\usepackage{tikz}
\usetikzlibrary{calc}
\usetikzlibrary{calc,trees,positioning,arrows,chains,shapes.geometric,fit,
    shapes,shadows,matrix}
\usetikzlibrary{shapes.arrows}
\usetikzlibrary{shapes.geometric,positioning}
\usepackage{enumerate}
\usepackage{tikz}

\newcommand{\h}{\hspace*{0.2in}}

\newtheorem{thm}{Theorem}
\newtheorem{lemma}{Lemma}
\newtheorem{definition}{Definition}

\begin{document}

\title{Necessary and Sufficient Conditions on Partial Orders for Modeling
Concurrent Computations}
%
%
%
%
%

\numberofauthors{2} 
%
\author{
%
%
\alignauthor Himanshu Chauhan\\
       \affaddr{The University of Texas at Austin}\\
       \email{himanshu@utexas.edu}
\and
\alignauthor Vijay K. Garg\\
       \affaddr{The University of Texas at Austin}\\
       \email{garg@ece.utexas.edu}
}

\newcommand{\dmp}{L_{DM} ( P )}
\newcommand{\lcgs}{L_{CGS}}
\newcommand{\lma}{L_{MA}}
\newcommand{\ldm}{L_{DM}}
\newcommand{\ra}{\rightarrow}
\newcommand{\figref}[1]{Fig. \ref{#1}}
\newcommand{\lt}{\leadsto}
\newcommand{\df}{\stackrel{\rm def}{=}}
\newcommand{\inx}[1]{#1\index{#1}}
\newcommand{\Ra}{\Rightarrow}
\newcommand{\pe}{\prec}
\newcommand{\peim}{\prec_{im}}
\newcommand{\mya}{\;\wedge\;}
\newcommand{\mlt}{<}
\newcommand{\sposet}{$(S,<,\tau)$}
\newcommand{\we}{width-extensible}
\newcommand{\scrf}{interleaving-consistent}
\newcommand{\xw}{$w$~}

\maketitle
\begin{abstract}
Partial orders are used extensively for modeling and analyzing concurrent computations. 
In this paper, 
 we define two properties of partially ordered sets:  
{\em width-extensibility} and {\em interleaving-consistency}, and show that 
a partial order can be a valid state based model:
(1) of some synchronous concurrent computation iff it is
 width-extensible, and (2) of some asynchronous
 concurrent computation iff it is width-extensible and interleaving-consistent.
We also show a duality between the event based and state based models of concurrent computations, and give
algorithms to convert models  between the 
two domains.
When applied to the problem of checkpointing, our theory leads to a better understanding 
of some existing results and algorithms in the field. 
It also leads to efficient detection algorithms for predicates whose evaluation requires knowledge of states
from 
all the processes in the system. 
\end{abstract}



\section{Introduction}
The `happened-before' relation introduced by Lamport
\cite{Lamp:HappenBefore} is a prevalent technique 
for modeling executions of distributed as well 
shared memory concurrent programs. 
The relation models causality and 
imposes a partial order on the set 
of events that occur in a computation. 
For a large number of applications,  models
based on events 
of the computation provide adequate basis for analysis. 
But for many applications such as global predicate detection \cite{GargWald:WeakUnstable} and checkpointing \cite{xu:zig}, it is beneficial to model a distributed computation
as a partial order on states of the 
involved processes.  
Events and states, however, are fundamentally different concepts. Events are instantaneous
and states have duration. A state captures values of all the variables
(including program counter) at a process, whereas an event captures the transition
of the system from one state to the other\footnote{Alternatively, one may model states
as instantaneous and events with duration. The point is that either the state or the event
must be modeled
with duration.}.
Although, there are multiple papers \cite{GargWald:WeakUnstable,God:1996:SV,flanagan2005dynamic} that model computations as partially ordered sets
(posets), there is no clear theory that brings out the distinction between
posets used for modeling event based executions and those 
used for modeling state based executions. This paper's first contribution is in 
establishing such a theory. 
For example, consider the posets in Fig.~\ref{fig:noncomp}. 
Are they valid
event based (or state based) models for some computation? 
What is the class of posets that characterize event based and state based models --- specifically, can every poset be a model for some
computation or there exist some restrictions on posets that model
the computations in event based or state based models? 
\remove{ }
Additionally, any model of a concurrent computation must
define the notion of a consistent
global state. Are the definitions different in state based and event based models?
One of the main goals of this paper is to establish results that form a basis to answer all these questions 
in a definitive manner. 
We study the relationship between the event based
models and state based models,
and characterize the exact class of posets that can be used to model
computations in either framework. We show a duality between the two models
that allows easy translation of algorithms from one model to the other.
\remove{
We model a computation in the event domain as a labeled poset
in which every event is labeled with the subset of (sequential) processes
on which that event is shared. The labeling of events must satisfy the
condition that all events that have a label for process $P_i$ must 
be totally ordered. The notion of shared events is similar to 
models such as CSP \cite{HoareCSP} and CCS \cite{Milner}.
In these computations, two or more processes may execute a shared event
resulting in state transition for them. For example, a distributed computation 
with blocking sends (in which the sender waits for the receiver to be ready to receive the
message) can be modeled using shared events.
We give an equivalent
state based model and show that a poset is a valid model for states of a 
concurrent computation {\em iff} it is width-extensible (defined in section \ref{sec:simple}). The poset in Fig.~\ref{fig:noncomp} is not
width-extensible and hence it is not a valid state based model.}
\remove{

Distributed computations in which communication  
between processes is achieved by asynchronous message 
passing are a special case of the general model of
concurrent computations. Using the 
event set notion, we can say that 
an {\em asynchronous distributed computation} is one in which the event set of different 
processes are disjoint.
Such a computation
is usually modeled by the Lamport's happened-before \cite{Lamp:HappenBefore} poset of events. 
In this paper, we define a class of posets called
width-extensible and interleaving-consistent and show that a poset is a valid model for states of an asynchronous
distributed computation {\em iff} it is width-extensible and interleaving-consistent. 

}
\remove{
The characterization of state based models for simple and complex computations is 
shown in Table \ref{tab:model}.

\begin{table}[htbp]
  \centering
  \begin{tabular}{0.95\textwidth}{| c | ll |}
    \hline
    {\bf Notion} & Event based Model & State based Model\\
    \hline
    Simple Computations &  Posets with no Shared Events & Width-extensible and Interleaving Consistent Posets\\
    Complex Computations &  Posets with Shared Events & Width-extensible Posets\\
    Consistent Global States &  Down-Sets (Order Ideals) & Width-Antichains\\
    \hline
  \end{tabular}
  \vspace{0.1in}
  \caption{State Event Duality
 }
   \label{tab:model}
\end{table}

}
In short, the key contributions of this paper are the following:
\begin{itemize}
\item we define two properties on posets: {\em width-extensibility}, 
and
{\em interleaving-consistency}, and 
show that they are necessary and sufficient 
conditions for posets modeling states of concurrent computations. 
\remove{
\item we establish the correspondence between event based and state based models 
of concurrent computations. 
\item we provide complete characterization, by establishing necessary and sufficient conditions, of partial orders that can model concurrent 
computations under the state based model.  
We prove that the class of width-extensible posets
forms a complete characterization of the state based models of concurrent computations,
and the class of width-extensible and interleaving-consistent posets
forms a complete characterization of the state based models of asynchronous distributed computations.
}
\item 
we give algorithms to translate event based models to state based models and vice-versa.
We establish the correspondence between the notions of the consistent global states in
these two models.
\item we show applications of our theory to the areas
of checkpointing and predicate detection (in Section
~\ref{sec:checkpoint}).   
\end{itemize}
\remove{
Some of the ideas and results of this paper may seem
intuitively 
straightforward, perhaps even obvious, but to the best
of our knoweldge they have not been presented in 
any published work. }

The rest of this paper is organized as follows.
Section~\ref{sec:back} covers the background concepts 
about modeling the concurrent computations as 
posets, and well-established concepts of event based models
of computations. 
Section~\ref{sec:simple} defines the state based models, 
and shows how to generate them from event based models. Sections~\ref{sec:synch} and~\ref{sec:adc} give complete characterization of state based models for synchronous and asynchronous 
concurrent computations. 
We conclude in 
Section~\ref{sec:checkpoint} by discussing 
the applications of our theory to the fields of checkpointing 
and predicate detection. 

\section{Background \& Terminology}
\label{sec:back}

We use the term {\em program} to represent a 
finite set of instructions, and {\em computation} to represent an execution of a program. In this paper, we 
restrict our focus to finite computations --- computations 
that terminate within bounded time. An event 
(of a computation) is a term that denotes --- 
depending on the context of the problem --- 
the execution 
of a single instruction or a collection of instructions
together.   
A concurrent computation  
is a computation  involving more than one
processes/threads --- it is possible that the 
instructions executed by different processes/threads are
different. Hence, a distributed computation is a concurrent 
 computation without shared memory
processes in which inter-process communication is possible 
only through message-passing. 
For modeling concurrent computations, the happened-before
relation ($\ra$) is defined as follows. The relation
$\ra$ on the set of events of a computation is the 
smallest relation that satisfies the following three 
conditions: $(1)$
If $a$ and $b$ are events in the same process 
and $a$ occurs before $b$, then $a \ra b$.  
$(2)$ For a distributed system, if $a$ is the sending
of a message and $b$ is the receipt of the same message, 
then $a \ra b$. For a shared memory system, if $a$ 
is the release of a lock by some thread and $b$ is the subsequent 
acquisition of that lock by any thread then $a \ra b$. $(3)$ If $a \ra b$ and $b \ra c$ then 
$a \ra c$.

Formally, a finite partially ordered set ({\em poset} in short) 
 is
a pair $P = (E, \ra)$ where $E$ is a finite set and $\ra$ is 
an irreflexive, antisymmetric, and
transitive binary relation on $E$ \cite{davey}.
We obtain a poset when we apply the happened-before ($\ra$) on the set of events of a
finite computation. Let $E$ be the 
set of events. Consider two events $a, b \in E$. 
If either $a \ra b$ or $b \ra a$, we say that $a$ and $b$ are {\em comparable}; otherwise,
we say $a$ and $b$ are {\em incomparable} or {\em concurrent} (in the 
context of concurrent computations), and denote this relation by $a~||~b$. Observe that $a~||~b~\wedge~b~||~c \not\Ra a~||~c$.   

It is important to note that 
multiple computations could have the identical posets as their model. 

\subsection{Concepts on Posets}

Let $P=(E,\ra)$ be a finite poset as defined above. 
A subset $Y\subseteq E$  is called an {\em chain (antichain)}, if every pair of distinct 
 points from $Y$ is comparable (incomparable) in $P$.
 The {\em height} of a poset
is defined to be the size of a largest chain in the poset.
 The {\em width} of a poset
is defined to be the size of a largest antichain in the poset.
All antichains of size equal to the width of the poset are called {\em width-antichains}
in this paper.
Let $\mathcal{A}(P)$ denote the set of all width-antichains
of $P$. Order $\leq$ is defined over $\mathcal{A}(P)$
as: \\
$A \leq B~~(A,B \in \mathcal{A}(P))$ iff $\forall a \in A, 
\exists b \in B: a \leq b$ in $P$. 

We model processes/threads 
as chains of posets, and thus 
events/states of every process/thread form a 
totally ordered chain.  
A family $\pi=(C_i~|~i=1,2\ldots,n)$ of chains of $P$ 
is called a chain partition of $P$ if $\bigcup{C_i~
|~i=1,2\ldots,n)}=P$. 

Given a subset $Y \subseteq E$, the {\em meet} of $Y$, if it exists, is the greatest lower bound
of $Y$ and the {\em join} of $Y$ is the least upper bound.
A poset $P = (X, \leq)$ is a {\em lattice} if joins and meets exist for all
finite subsets of $X$.
Let $P$ be a poset with a given chain partition of width $w$.
In a concurrent computation, $P$ is the set of events
executed under the happened-before partial order. Each chain would correspond to a
total order of events executed on a single process.
In such a poset, every element $e$ can be identified with a tuple
$(i,k)$ which represents the $k$th event in the $i$th process; $1 \leq i \leq w$.

A subset $Q$ is a {\em downset} (also called {\em order ideal}), of $P$ if it satisfies the
constraint that if $f$ is in $Q$ and $e$ is less than or equal to $f$, then
$e$ is also in $Q$. 
When a computation is modeled
as a poset of events, the downsets are called {\em consistent cuts}, or
{\em consistent global states} \cite{ChanLamp:Snap}. Throughout 
this paper, we use the term consistent cut. 
The set of downsets is closed under both union and intersection and
therefore forms a lattice under the set containment order \cite{davey}.

\subsection{Event based Model of Concurrent Computations}

As discussed earlier, a concurrent computation is usually modeled as a set of
events, $E$, together with a partial order {\em happened-before} \cite{Lamp:HappenBefore}, denoted by $\ra$.
Implicit in this model is the partition of $E$
into chains corresponding to the processes 
on which the events are executed. This partition is called a 
{\em chain partition}. 
We make this partition explicit in our model because the translation
of the event based model into the state based model depends upon
it. 

\begin{definition}[Event based model of computation]
\label{def:event_model}
A concurrent computation on $n$ processes is modeled by 
$\hat{E}=(E, \ra, \pi)$, where
$E$ is the set of events, $\ra$ is the happened-before
relation on $E$,
and $\pi$ maps every event to a subset of processes from $\{1..n\}$ such that
for all $i \in \{1..n\}: E_i = \{e \in E~|~i \in \pi(e) \}$ is totally ordered under $\ra$.
\end{definition}

Here, $\pi$ is a chain partition of poset 
defined by $(E,\ra)$. 
Intuitively, in the context of concurrent 
computations, $\pi$ maps events executed on a single process to a 
total order such that $E_i$ is the totally 
ordered set of events executed on process/thread $C_i$. Note that an event, such as execution of a barrier, could be assigned to multiple processes. If an event $e \in E_i \cap E_j$, then $e$ is a `shared' event for processes $C_i$ and $C_j$.
\tikzstyle{place}=[circle,draw=black,thick,inner sep=0pt,minimum size=2mm]
\tikzstyle{nsplace}=[circle,draw=black,fill=black,thick,inner sep=0pt,minimum size=2mm]
\tikzstyle{place1}=[circle,draw=black,thick,inner sep=0pt,minimum size=3mm]
\tikzstyle{inner}=[circle,draw=black,fill=black,thick,inner sep=0pt,minimum size=1.25mm]
\tikzstyle{satblock} = [rectangle, draw=gray, thin, fill=black!20,
text width=3.5em, text centered, rounded corners, minimum height=2em]
\tikzstyle{nsatblock} = [rectangle, draw=gray, thin,
text width=3.5em, text centered, rounded corners, minimum height=2em]
\tikzstyle{psatblock} = [rectangle, draw=gray, thin,
text width=3.85em, text centered, rounded corners, minimum height=2em]
\begin{figure*}[!htb]
\centering
\begin{subfigure}[b]{0.3\textwidth}
\begin{tabular}{l|l}
Proc.~$1$ & Proc.~$2$ \\
\hline
$1$: local event ($a$) & $1$: local event ($e$)\\
$2$: send msg ($b$) & $2$: receive msg ($f$)\\
$3$: local event ($c$) & $3$: local event ($g$)\\
\end{tabular}
\caption{Pseudocode of instructions}
\label{fig:simple_instr}
\end{subfigure}
\qquad
\begin{subfigure}[b]{0.3\textwidth}
\centering
\begin{tikzpicture}[]
\node at ( 0,1) [nsplace] [label=below:$a$] (a) {};
\node at ( 1,1) [nsplace] [label=below:$b$] (b) {};
\node at ( 2.5,1) [nsplace] [label=below:$c$] (c) {};
\node at ( 0,0) [nsplace] [label=below:$e$] (e) {};
\node at ( 1.5,0) [nsplace] [label=below:$f$](f) {};
\node at ( 3,0) [nsplace] [label=below:$g$] (g) {};
\draw [thick] (b.west) -- (a.east);
\draw [thick] (c.west) -- (b.east);
\draw [thick] (f.west) -- (e.east);
\draw [thick] (g.west) -- (f.east);

\draw [thick,->] (a.east) -- (b.west);
\draw [thick,->] (b.east) -- (c.west);
\draw [thick,->] (e.east) -- (f.west);
\draw [thick,->] (f.east) -- (g.west);
\draw [thick,->] (b) -- (f);
\end{tikzpicture}
\caption{Event Based Model}
\label{fig:run-event}
\end{subfigure}
\qquad
\begin{subfigure}[b]{0.3\textwidth}
\begin{tikzpicture}[]
\node at ( -1,1) [place] [label=below:$a_0$] (a0) {};
\node at (-0.5,1) [] [label={above:\textcolor{blue}{$a$}}] () {};
\node at ( 0,1) [place] [label=below:$a'$] (a) {};
\node at (0.5,1) [] [label={above:\textcolor{blue}{$b$}}] () {};
\node at ( 1,1) [place] [label=below:$b'$] (b) {};
\node at (1.75,1) [] [label={above:\textcolor{blue}{$c$}}] () {};
\node at ( 2.5,1) [place] [label=below:$c'$] (c) {};
\node at ( -1,0) [place] [label=below:$e_0$] (e0) {};
\node at (-.5,-.1) [] [label={above:\textcolor{blue}{$e$}}] () {};
\node at ( 0,0) [place] [label=below:$e'$] (e) {};
\node at (.75,-.1) [] [label={above:\textcolor{blue}{$f$}}] () {};
\node at ( 1.5,0) [place] [label=below:$f'$](f) {};
\node at (2.25,-.1) [] [label={above:\textcolor{blue}{$g$}}] () {};
\node at ( 3,0) [place] [label=below:$g'$] (g) {};
\draw [thick,->] (a0.east) -- (a.west);
\draw [thick,->] (a.east) -- (b.west);
\draw [thick,->] (b.east) -- (c.west);
\draw [thick,->] (e0.east) -- (e.west);
\draw [thick,->] (e.east) -- (f.west);
\draw [thick,->] (f.east) -- (g.west);
\draw [thick,->] (a) -- (f);
\end{tikzpicture}
\caption{State Based Model}
\label{fig:run-state}
\end{subfigure}
\caption{An example distributed computation and its patial order models}
\label{fig:simple}
\end{figure*}
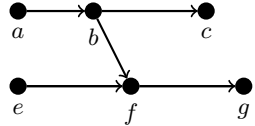
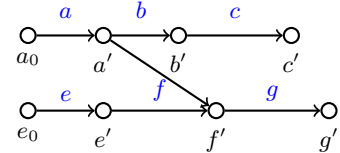
\tikzstyle{place}=[circle,draw=black,fill=white,thick,inner sep=0pt,minimum size=2mm]
\tikzstyle{blackCirc}=[circle,draw=black,fill=black,thick,inner sep=0pt,minimum size=2mm]
\tikzstyle{nsplace}=[]
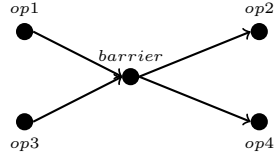
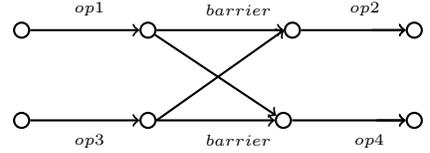
\begin{figure*}[htbp]
\begin{subfigure}[b]{0.3\textwidth}
\begin{tabular}{l|l}
Proc.~$1$ & Proc.~$2$ \\
\hline
$1$: local event ($op1$) & $1$: local event ($op3$)\\
$2$:  execute $barrier$ & $2$:  execute $barrier$\\
$3$: local event ($op2$) & $3$: local event ($op4$)\\
\end{tabular}
\caption{Pseudocode of instructions}
\label{fig:barrier_instr}
\end{subfigure}
\qquad
\begin{subfigure}[b]{0.3\textwidth}
\centering
{\scalefont{0.7}
\begin{tikzpicture}[scale=0.6]
\node at (0,2) [blackCirc] [label={above:$op1$}] (a) {};
\node at (2.35,1) [blackCirc] [label={above:$barrier$}] (b) {};
\node at (5.2,2) [blackCirc] [label={above:$op2$}] (c) {};
\node at (0,0) [blackCirc] [label={below:$op3$}] (e) {};
\node at (5.2,0) [blackCirc] [label={below:$op4$}] (g) {};
\draw [thick,->] (a.east) -- (b.west);
\draw [thick,->] (b.east) -- (c.west);
\draw [thick,->] (e.east) -- (b.west);
\draw [thick,->] (b.east) -- (g.west);
\end{tikzpicture}
}
\caption{Event based model}
\label{fig:barrier-run-event}
\end{subfigure}
\begin{subfigure}[b]{0.3\textwidth}
{\scalefont{0.7}
\begin{tikzpicture}[scale=0.6]
\node at (0,2) [nsplace] [label={above:{$op1$}}] (a) {};
\node at (1.3,2) [place] [] (a1) {};
\node at (1.3,0) [place] [] (onefour) {};
\node at (3.3,2) [nsplace] [label={above:{$barrier$}}] (b) {};
\node at (6.1,2) [nsplace] [label={above:{$op2$}}] (c) {};
\node at (4.5,2) [place] [] (c1) {};
\node at (0,0) [nsplace] [label={below:{$op3$}}] (e) {};
\node at (3.3,0) [nsplace] [label={below:{$barrier$}}](f) {};
\node at (4.3,0) [place] [](f1) {};
\node at (6.2,0) [nsplace] [label={below:{$op4$}}] (g) {};
\node at (7.2,2) [place] [] (end1) {};
\node at (7.2,0) [place] [] (end2) {};
\draw [thick,->] (c.east) -- (end1.west);
\draw [thick,->] (g.east) -- (end2.west);
\node at (-1.5,2) [place] [] (sonezero) {};
\node at (-1.5,0) [place] [] (stwozero) {};
\node at (-2,2) [] (bot11) {};
\node at (-2,0) [] (bot22) {};
\draw [thick,->] (sonezero.east) -- (a1.west);
\draw [thick,->] (onefour.east) -- (c1.west);
\draw [thick,->] (a1.east) -- (c1.west);
\draw [thick,->] (c1.east) -- (end1.west);
\draw [thick,->] (stwozero.east) -- (onefour.west);
\draw [thick,->] (onefour.east) -- (f1.west);
\draw [thick,->] (f1.east) -- (end2.west);
\draw [thick,->] (a1) -- (f1);
\end{tikzpicture}
}
\caption{State based model}
\label{fig:barrier-run-state}
\end{subfigure}
\caption{A computation with a barrier and its partial order models}
\label{fig:barrier}
\end{figure*} 
Fig.~\ref{fig:simple}\subref{fig:run-event} shows the
event based model of a distributed computation resulting from the
 execution of the pseudocode instructions listed in 
Fig.~\ref{fig:simple}\subref{fig:simple_instr}. 
Fig.~\ref{fig:barrier}\subref{fig:barrier-run-event} shows the 
event based model of a concurrent computation on two processes 
that synchronize using a barrier (as per the instructions listed
in Fig.~\ref{fig:barrier}\subref{fig:barrier_instr}). 
Note that the model of Fig.~\ref{fig:barrier}\subref{fig:barrier-run-event} allows us to represent synchronous messages where the sender
blocks for the receiver to be ready. Such synchronous messages are represented by a
single event $e$ such that $\pi(e)$ includes the sender as well as the receiver.
The model also allows us to represent barriers which require multiple processes
to wait until all the processes participating in the barrier execute it.
It can also model behavior of finite communicating sequential processes \cite{HoareCSP}.

{\bf Note}: In all the figures throughout this paper, events are depicted with dark 
filled circles, and states are depicted with empty circles.

Generally, the analysis of concurrent computations 
requires reasoning over the valid states of the system that could 
occur in these computations. These states are commonly
called {\em consistent global states} or {\em consistent cuts}. 
\begin{definition}[Consistent cut in event based model]
\label{def:cgs_event}
Given an event based model $(E, \ra, \pi)$ of 
a computation, $G \subseteq E$ is a {\em consistent cut} of the
computation if
$~\forall e,f \in E:   (f \in G) \wedge (e \ra f) \Rightarrow (e \in G)$. 
\end{definition}
Note that this definition is independent of $\pi$ and coincides
with the definition of a {\em down-set} of a poset 
\cite{davey}.
It is well known that the set of downsets forms a distributive 
lattice.
Conversely, Birkhoff showed that every finite distributive lattice
can be generated as the set of downsets of a poset \cite{Birk4}.
Thus, finite distributive lattices completely characterize the set of consistent cuts
in the event based model.

The consistent cuts of the 
event based model 
in
Fig.~\ref{fig:run-event} are: 
$\{\},\{a\},\{e\},\{a,b\},\{a,e\},\{a,b,c\},\{a,b,e\},\{a,b,c,e\},$\\
$\{a,b,e,f\},\{a,b,c,e,f\},\{a,b,e,f,g\},\{a,b,c,e,f,g\}$. 

\section{Modeling Computations using 
\\States}
\label{sec:simple}

For many applications in concurrent debugging 
\cite{tzoref2007}, and 
predicate detection in distributed systems it is more natural to
model a computation using states rather than events.
For example, we may be interested in the cut (global state) in which
all processes have taken their local checkpoint. We first 
give an intuition for 
state based model of concurrent computations. 
%
%
An event is always 
executed in some state, and the state before the 
event's execution `existed-before' the state 
resulting from the execution. The existed-before relation
between states is denoted using ``$<$''. 
The diagram (denoting the happened-before relation) of the model based on events in Fig.~\ref{fig:simple}\subref{fig:run-event} 
corresponds to the state based model shown in
Fig.~\ref{fig:simple}\subref{fig:run-state}.
In this figure, the execution of event $a$ gets translated into 
an edge between two states: initial state $a_0$ (that
existed before $a$ was executed), 
 and state $a'$ (the state immediately after $a$'s execution). Thus, 
we have $a_0 < a'$ in the state based model.  


 Although some concepts 
carry over from events
to states, there are some important differences. For example, any poset of events in which all events on a single process are totally ordered can be a model of some 
concurrent computation in the happened-before model. But, not every
poset of states is a valid concurrent computation. Consider the poset  in 
\figref{fig:invalidB}. If this poset were to 
be used as a state based model of a computation, 
the model would be incorrect ---   
because even if the modeled states form a poset, the
equivalent event based model  
would have a cycle (as shown in Fig.~\ref{fig:eventcycle})\footnote{The techniques
involved in generating event based model 
from state based model are in the next section.}.
 Thus, we can allow only those partial orders on states that
do not induce cycles on the order on events.

We claim that a poset can only be a valid state based model of a concurrent computation if it satisfies a notion
called 
{\em width-extensibility}. 

\begin{definition}[Width-extensible Poset]
A poset $(X, <)$ is width-extensible if and only if for every
antichain $A \subseteq X$, there exists a width-antichain
$W$ containing $A$.
\end{definition}

\tikzstyle{place}=[circle,draw=black,fill=white,thick,inner sep=0pt,minimum size=2mm]
\tikzstyle{nsplace}=[]
\begin{figure}[!htb]
\begin{subfigure}[b]{0.4\textwidth}
\centering
\begin{tikzpicture}[scale=0.7]
\node at ( 0,2) [place] [label={above:$a$}] (a) {};
\node at (3,2) [place] [label={above:$b$}] (b) {};
\node at (6,2) [place] [label={above:$c$}](c) {};
\node at (0,0) [place] [label={above:$d$}] (e) {};
\node at (4,0) [place] [label={above:$e$}](f) {};
\draw [thick,->] (a.east) -- (b.west);

\draw [thick,->] (e.east) -- (f.west);
\draw [thick,->] (b.east) -- (c.west);
\draw [thick,->] (b) -- (f);
\draw [thick,->] (e) -- (b);
\end{tikzpicture}
\caption{Not width-extensible: no width-antichain for $\{b\}$}
\label{fig:invalidB}
\end{subfigure}
\begin{subfigure}[b]{0.4\textwidth}
\centering
\begin{tikzpicture}[scale=0.7]
\node at ( 0,1.4) [place] [label={above:$a$}] (a) {};
\node at (2,1.4) [place] [label={above:$b$}] (b) {};
\node at (4,1.4) [place] [label={above:$c$}] (c) {};
\node at (0,0) [place] [] [label={above:$d$}] (d) {};
\node at (2,0) [place] [] [label={above:$e$}] (e) {};
\node at (4,0) [place] [] [label={above:$f$}] (f) {};

\node at (0,-1.4) [place] [] [label={above:$g$}] (g) {};
\node at (2,-1.4) [place] [] [label={above:$h$}] (h) {};
\node at (4,-1.4) [place] [] [label={above:$i$}] (i) {};

\draw [thick,->] (a.east) -- (b.west);
\draw [thick,->] (b.east) -- (c.west);

\draw [thick,->] (d.east) -- (e.west);
\draw [thick,->] (e.east) -- (f.west);

\draw [thick,->] (g.east) -- (h.west);
\draw [thick,->] (h.east) -- (i.west);
\draw [thick,->] (b) -- (f);
\draw [thick,->] (e) -- (i);
\end{tikzpicture}
\caption{Not width-extensible: no width-antichain for $\{b,i\}$}
\label{fig:noantichain}
\end{subfigure}
\caption{Invalid posets under the state based model}
\label{fig:noncomp}
\end{figure}
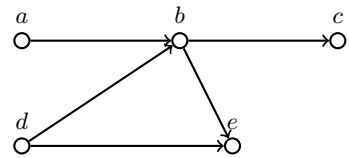
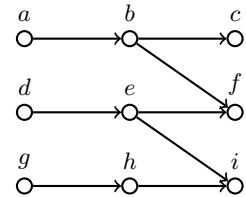
Informally, when states of a concurrent computation are modeled as a poset,
this property requires that for any set of incomparable local states there
is a possible consistent cut that includes these local states.
We will show later that in the state based model, the consistent cuts
correspond to width-antichains (and not down-sets).
The poset in Fig.~\ref{fig:invalidB} is not width-extensible  
because there is no width-antichain that contains $b$.

In the above definition of width-extensible posets, we can not substitute
``for all antichains'' by ``for all antichains of size $1$''.
In the example of Fig.~\ref{fig:noantichain}, there is a width-antichain for every individual element $a$ to $i$. 
This can be easily verified as 
$\{a,d,g\}$, $\{b,e,h\}$, and $\{c,f,i\}$ are all width-antichains. 
But there is no width-antichain that contains $\{b,i\}$.
Hence, the poset is not width-extensible. 

We now show a surprising 
result: it is sufficient to restrict our
attention to antichains of size two for checking 
width-extensibility. 
\begin{thm}
A poset $(X,<)$ is width-extensible if and 
only if for every antichain $A$ of size
at most two, there exists a width-antichain $W$ containing $A$.
\end{thm}
\begin{proof}
The necessity is obvious --- 
 the definition 
of width-extensibility demands that every antichain 
is contained in some width-antichain. Hence, 
for $(X,<)$ to be width-extensible, antichains 
of size at most two must also be contained in 
a width-antichain. 
We now prove sufficiency.
We want to prove that
if every antichain of size 
at most two is contained in a width-antichain, then every antichain (of any size)
is also contained in a width-antichain. Let \xw be the width of the poset $(X,<)$ and
$\{C_1,C_2,...,C_w\}$ be a chain partition of size $w$. Consider an antichain $A$ of size $k, 3 \leq k \leq w$. If \xw $= k$, then $A$ itself 
is a width-antichain, and we have the result.
Suppose \xw $> k$, and $A$ is not contained in any width-antichain. Hence, there is some chain  $C_i$  such that $A$ does not have any elements from $C_i$. We know that for any pair of elements $a, b \in A$, with $a \neq b$, the antichain $\{a, b\}$ is width-extensible. Let $I_i(a,b)$ denote the maximal interval on $C_i$ that contains 
all the elements that are incomparable to both $a$ and $b$. As $\{a, b\}$ is width-extensible, we know that $I_i(a,b)$ is non-empty.  
Now consider $a,b,c \in A$, where all three are distinct. The width-extensibility of size two antichains 
guarantees 
that $I_i(a,b)$, $I_i(b,c)$, and $I_i(a,c)$ are 
all non-empty. Since every pair of these intervals have non-empty intersection, and all intervals are
sets of one or more consecutive states in $C_i$, we get that
$I_i(a,b) \cap I_i(b,c) \cap I_i(a,c) \neq \phi$.
 This means that $\exists d \in C_i : (d~||~a) \wedge (d~||~b) \wedge (d~||~c)$, i.e. $d$ is concurrent to $a$, $b$, 
 and $c$. Hence, $d$ can be added to $A$.
By repeating this argument for all chains that do not have any element in $A$, we can extend $A$ to a width-antichain.
\end{proof}
We can now define the state based model of a concurrent computation as follows: 
\begin{definition}[State based model of concurrent computations]
\label{def:sbcc}
A concurrent computation on $n$ processes is modeled 
by $\hat{S}$: a tuple $(S, < , \tau)$, where
$S$ is the set of local states, $(S, <)$ is a width-extensible poset,
and $\tau$ is a map from $S$ to $\{1..n\}$ such that 
for all distinct states $s,t \in S$
for all $i \in \{1..n\},  S_i = \{s \in S~|~i \in \tau(s) \}$ is totally ordered under $<$. 
i.e.,~~ 
$ \tau(s) = \tau(t) \Ra (s < t) \vee (t < s)$. 
\end{definition}

Thus, $\tau$ partitions $S$ such that every block of the partition $S_i$ is totally
ordered. The relation $<$ between states captures the
 `existed-before' notion discussed in the first para
 of Section~\ref{sec:simple}. 
Fig.~\ref{fig:simple}\subref{fig:run-state} and  \ref{fig:barrier}\subref{fig:barrier-run-state}, are corresponding state
based models of event based models shown in 
Fig.~\ref{fig:simple}\subref{fig:run-event} and  \ref{fig:barrier}\subref{fig:barrier-run-event}. Note that 
in these figures (of state based models), the events 
 are shown as edge labels above the edges
 that capture $<$ (existed-before) relation on the states. 

We now show the difference in the definitions of consistent cuts in the
state based and event based model.
\begin{definition}[Consistent cut in state based model]
\label{def:cgs_state}
Under the state based model, $(S,<,\tau)$, of a concurrent computation ,
a subset $T \subseteq S$ of size equal to the width of poset
$(S,<)$ is a consistent
cut if $~\forall s,t \in
T: s~||~t$. 
\end{definition}
The order ``$<$'' over consistent 
cuts is defined using the ``$\leq$'' relation defined 
over width-antichains in Section~\ref{sec:back}. 
Under the state based model, for any two consistent cuts $A, B$ 
we have: 
$A < B$ iff $A \leq B \wedge A \neq B$. Hence, 
$A < B \Rightarrow  \exists a \in A, \exists b \in B: a < b$ in $(S,<)$.
It is clear that the consistent
cuts in state based model correspond to width-antichains of the poset.

The consistent cuts of the state based model of  Fig.~\ref{fig:simple}\subref{fig:run-state} are: 
$\{a_0,e_0\},\{a',e_0\},
\{a_0,e'\},\{b',e_0\},\{a',e'\},
\{c',e_0\},$\\
$\{b',e'\},\{c',e'\},\{b',f'\},$
$\{c',f'\},\{b',g'\},\{c',g'\}$.
%

At this point we have two notions of a consistent cut of a concurrent computation: one in the
event based model (Defn.~\ref{def:cgs_event}) and the
other in the state based model 
(Defn.~\ref{def:cgs_state}). 
Dilworth \cite{Dilworth50} proved that the set of all width-antichains also forms a distributive lattice, and Koh \cite{Koh} showed that every finite
distributive lattice can be generated as the set of width-antichains of a poset. The lattice of width-antichains is in general a sublattice of the
lattice of downsets. Thus, the notion of consistent global states is different in event based and state based models, a distinction that has not been explored in distributed computing literature.
It is also important to question that what is the relationship between these two definitions?
In the next section, we show that there is a `one-to-one' correspondence between consistent cuts in the event based
and the state based models. 

\remove{
\begin{proof}
Let $G$ be any consistent cut of $(E, \ra, \pi)$. We will show how to construct the corresponding
consistent cut $T$ of \sposet.
Suppose that $G$ contains at least one event from $C_i$.
Then, let $(i,k)$ be the largest event from process $C_i$. In this case, we add $[i,k]$ to $T$.
If $G$ does not contain any event from $C_i$, then we add $[i,0]$ to $T$.
Clearly, $T$ has exactly $n$ states, one from each process.  We show that the cut $T$ is also consistent.
If not, suppose $[i,s]$ and $[j,t]$ be two states in $T$ such that $[i,s] < [j,t]$. This implies that
$(i,s+1) \ra (j,t)$, under the event based model, contradicting that $G$ is consistent because $G$ contains $(j,t)$ but does not contain $(i,s+1)$.
It is also easy to verify that the mapping from the set of consistent cuts is one-to-one.

Conversely, given a consistent cut $T$ 
in the state based model, we construct a consistent 
cut in event based model in $1-1$ manner as follows. For all states $[i,k] \in T$ we include
all events $(i,k')$ such that $k ' \leq k$. Note that when $k$ equals $0$, no events from $C_i$ are included.
It can again be easily verified that whenever $T$ is a consistent 
cut in state model, $G$ is a consistent cut in event model.
\end{proof}
Let $G$ be any consistent cut of $(E, \ra, \pi)$. We will show how to construct the corresponding
consistent cut $T$ of \sposet.
Suppose that $G$ contains at least one event from $C_i$.
Then, let $(i,k)$ be the largest event from process $C_i$. In this case, we add $[i,k]$ to $T$.
If $G$ does not contain any event from $C_i$, then we add $[i,0]$ to $T$.
Clearly, $T$ has exactly $n$ states, one from each process.  We show that the cut $T$ is also consistent.
If not, suppose $[i,s]$ and $[j,t]$ be two states in $T$ such that $[i,s] < [j,t]$. This implies that
$(i,s+1) \ra (j,t)$, under the event based model, contradicting that $G$ is consistent because $G$ contains $(j,t)$ but does not contain $(i,s+1)$.
It is also easy to verify that the mapping from the set of consistent cuts is $1-1$.

Conversely, given a consistent cut $T$ 
in the state based model, we construct a consistent 
cut in event based model in $1-1$ manner as follows. For all states $[i,k] \in T$ we include
all events $(i,k')$ such that $k ' \leq k$. Note that when $k$ equals $0$, no events from $C_i$ are included.
It can again be easily verified that whenever $T$ is a consistent cut, $G$ is a consistent prefix.
\end{proof}
}

\subsection{Translation between event based and state based models}
\label{subsec:translate}

Let $\hat{E} = (E, \rightarrow, \pi)$ be an 
event based model of a computation on $n$ processes/threads.
Let $\pi$ partition $E$ into $n$ chains: ($E_i ~|~ i = 1, 2, \ldots n$).
For each
$i= 1, 2, \ldots n$, let $|E_i| = n_i (\geq 1)$. Suppose the elements of $E_i$ are named as
follows:
$E_i : (i, 1)\rightarrow(i, 2)\rightarrow \ldots \rightarrow(i, n_i - 1)\rightarrow(i, n_i)$. 
Note that if an event is `shared' between two processes $i$ and $j$, then it will have 
two labels $(i, x)$ and $(j, y)$, with $1 \leq x \leq n_i$, and  $1 \leq y \leq n_j$.\footnote{By extension of this rule, 
an event that is `shared' between $k$ processes would 
have $k$ labels.} 
We generate a state based model $\hat{S} = (S, <, \tau)$
from $\hat{E}$ using the following  
function. \\
\\
\fbox{
\parbox{0.45\textwidth}{
{\bf Function $ES$ Transform}: \\
For each  $i= 1, 2, \ldots n$, let $S_i$ be an $|n_i + 1|$ element chain where $n_i = |E_i|$ as above. Define the elements in $S_i$ as follows: 
\[
   S_i:[i,0] < [i,1] < ... < [i,n_i - 1] < [i,n_i].
\]
Let $S = \bigcup_{i=1}^{n}  S_i$ and define a binary relation ``$<$''
on $S$ by putting
$ [i,r] < [j,s] $ in $S$ $(i,j = 1,2,\ldots,n; 0\leq r\leq n_i, 0\leq s
\leq n_j)$~ {\em iff}:
\begin{itemize}\itemsep0pt
\item $r < s$, if $i = j$ 
\item $(i,r + 1)$ and $(j,s)$ are both present in $E$
\item and $(i,r + 1) < (j,s)$ in $E$ if $i\neq j$.
\end{itemize}
}
}

A special case of this transform, on disjoint 
chain partitions, was used by Koh in \cite{Koh} to prove properties of lattice
of width-antichains. 

Fig.~\ref{fig:koh_transform} gives illustrations of the application of
this transform.

In the generated $\hat{S}$, $\tau$ is dependent on 
the chain partition 
$\pi$ in $\hat{E}$. Intuitively, every state chain $S_i$ 
contains the states of process $i$, such 
that event $(i, k)$ in $E_i$, here $1 \leq k \leq n_i$, causes a transition from state $[i, k-1]$
to $[i, k]$. 
On chain $S_i$, the state $[i, 0]$ represents the initial state of the process $i$, and $[i, n_i]$ represents the final state of the process $i$.
  The worst-case complexity of the $ES$ transform
is $\mathcal{O}(|E|^2)$.

\tikzstyle{place}=[circle,draw=black,fill=white,thick,inner sep=0pt,minimum size=2mm]
\tikzstyle{nsplace}=[circle,fill=black,draw=black,thick,inner sep=0pt,minimum size=2mm]
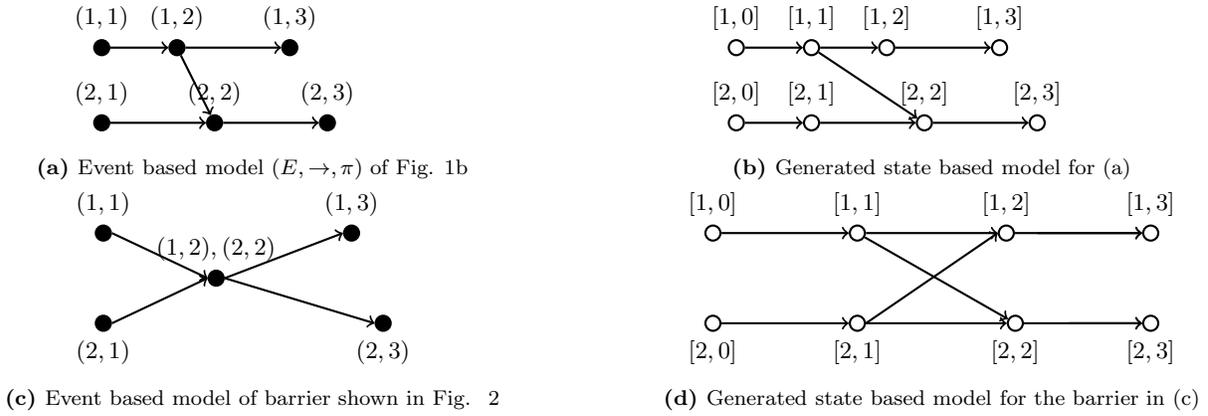
\begin{figure*}[!htb]
\centering
\begin{subfigure}[b]{0.49\textwidth}
\centering
\begin{tikzpicture}
\clip (-1,-0.25) rectangle (5,1.6) ;
\node at ( 0,1) [nsplace] [label={above:$(1,1)$}] (a) {};
\node at ( 1,1) [nsplace] [label={above:$(1,2)$}] (b) {};
\node at ( 2.5,1) [nsplace] [label={above:$(1,3)$}] (c) {};
\node at ( 0,0) [nsplace] [label={above:$(2,1)$}] (e) {};
\node at ( 1.5,0) [nsplace] [label={above:$(2,2)$}] (f) {};
\node at ( 3,0) [nsplace] [label={above:$(2,3)$}] (g) {};
\draw [thick, ->] (a.east) -- (b.west);
\draw [thick, ->] (b.east) -- (c.west);

\draw [thick, ->] (e.east) -- (f.west);
\draw [thick, ->] (f.east) -- (g.west);

\draw [thick,->] (b) -- (f);
\end{tikzpicture}
\caption{Event based model $(E,\ra,\pi)$ of~\figref{fig:run-event}}
\label{fig:ex_chains}
\end{subfigure}
\hfill
\begin{subfigure}[b]{0.49\textwidth}
\centering
\begin{tikzpicture}
\clip (-1.8,-0.25) rectangle (5,1.6) ;
\node at ( -1,1) [place] [label={above:$[1,0]$}] (a0) {};
\node at ( 0,1) [place] [label={above:$[1,1]$}] (a) {};
\node at ( 1,1) [place] [label={above:$[1,2]$}] (b) {};
\node at ( 2.5,1) [place] [label={above:$[1,3]$}] (c) {};
\node at ( -1,0) [place] [label={above:$[2,0]$}] (e0) {};
\node at ( 0,0) [place] [label={above:$[2,1]$}] (e) {};
\node at ( 1.5,0) [place] [label={above:$[2,2]$}] (f) {};
\node at ( 3,0) [place] [label={above:$[2,3]$}] (g) {};

%
\node at (3.5, 1) [] (end1) {};
\node at (3.5, 0) [] (end2) {};

\draw [thick, ->] (a.east) -- (b.west);
\draw [thick, ->] (a0.east) -- (a.west);
\draw [thick, ->] (e0.east) -- (e.west);
\draw [thick, ->] (b.east) -- (c.west);

\draw [thick, ->] (e.east) -- (f.west);
\draw [thick, ->] (f.east) -- (g.west);

\draw [thick,->] (a) -- (f);
\end{tikzpicture}
\caption{Generated state based model for 
 (a)}
\label{fig:sppi}
\end{subfigure}
\begin{subfigure}[b]{0.49\textwidth}
\centering
\begin{tikzpicture}[scale=0.6]
\node at ( 0,2) [blackCirc] [label={above:$(1,1)$}] (a) {};
\node at (2.5,1) [blackCirc] [label={above:$(1,2), (2,2)$}] (b) {};
\node at (5.5,2) [blackCirc] [label={above:$(1,3)$}] (c) {};
\node at (0,0) [blackCirc] [label={below:$(2,1)$}] (e) {};
\node at (6.2,0) [blackCirc] [label={below:$(2,3)$}] (g) {};
\node at (7.2, 2) [] (end1) {};
\node at (7.2, 0) [] (end2) {};
\draw [thick,->] (a.east) -- (b.west);
\draw [thick,->] (b.east) -- (c.west);

\draw [thick,->] (e.east) -- (b.west);
\draw [thick,->] (b.east) -- (g.west);
\end{tikzpicture}
\caption{Event based model of barrier shown in Fig.~
\ref{fig:barrier}}
\label{fig:barrier-koh-event}
\end{subfigure}
\hfill
\begin{subfigure}[b]{0.49\textwidth}
\centering
\begin{tikzpicture}[scale=0.6]
\node at (1.7,2) [place] [label={above:$[1,1]$}] (a1) {};
\node at (1.7,0) [place] [label={below:$[2,1]$}] (onefour) {};
\node at (5,2) [place] [label={above:$[1,2]$}] (c1) {};
\node at (5.2,0) [place] [label={below:$[2,2]$}](f1) {};
\node at (8.2,2) [place] [label={above:$[1,3]$}] (end1) {};
\node at (8.2,0) [place] [label={below:$[2,3]$}] (end2) {};
\draw [thick,->] (c.east) -- (end1.west);
\draw [thick,->] (g.east) -- (end2.west);
\node at (-1.5,2) [place] [label={above:$[1,0]$}] (sonezero) {};
\node at (-1.5,0) [place] [label={below:$[2,0]$}] (stwozero) {};
\draw [thick,->] (sonezero.east) -- (a1.west);
\draw [thick,->] (onefour.east) -- (c1.west);
\draw [thick,->] (a1.east) -- (c1.west);
\draw [thick,->] (c1.east) -- (end1.west);
\draw [thick,->] (stwozero.east) -- (onefour.west);
\draw [thick,->] (onefour.east) -- (f1.west);
\draw [thick,->] (f1.east) -- (end2.west);
\draw [thick,->] (a1) -- (f1);
\end{tikzpicture}
\caption{Generated state based model for 
the barrier in (c)}
\label{fig:barrier-koh-state}
\end{subfigure}
\caption{Event to State transform for computations of earlier examples}
\label{fig:koh_transform}
\end{figure*}
We show that this $\hat{S}$, generated by applying the 
$ES$ transform on $\hat{E}$, is a valid state based model 
of the concurrent computation, i.e., it is a 
width-extensible poset.
We first show that it is a poset. 
\begin{lemma}
\label{lem:poset}
If $\hat{S}$ is the result of applying $ES$ transform on an event based
model $\hat{E}=(E,\ra,\pi)$ of a concurrent computation then
$\hat{S}$ is a poset under the ``$<$'' relation.
\end{lemma}
\begin{proof}
We show that the relation ``$<$'' on $S$ is transitive and antisymmetric, and thus irreflexive. 
\begin{itemize}
\item {\em Claim (i)} The relation ``$<$" is asymmetric. \\
{\em Proof}:
Let $[i, r],[j, s] \in S$ such that $[i, r]<[j, s]$. Clearly, $[j, s] \not< [i, r]$ if $i=j$; otherwise we would get
 $s \ra r$ in $E$. Assume $i \neq j$ and
$[j, s]<[i, r]$. Then by definition, we have $(i,r+1)\ra(j,s)\ra 
(j, s + 1)\ra(i, r)$ in $E$, which is impossible as it violates
the asymmetry of $\ra$ in $E$. 
\item {\em Claim (ii)} The relation ``$<$" is transitive. \\
{\em Proof}:
Let $[i, r], [j, s],[k, t] \in \hat{S}$ such that $[i, r]<[j, s]$ and $[j, s]<[k, t]$. Assume $i \neq j$
and $k\neq j$. Then we have $(i,r+1) \ra (j,s) \ra (j,s+1) \ra (k, t)$ and hence $(i,r+1) \ra 
(k, t)$ in $E$, which implies that $[i, r]<[k, t]$ whether $i = k$ or $i \neq k$. The cases for
$i= j$ or $j = k$ can be proved similarly.
\end{itemize}
Hence $\hat{S}$ forms a poset under the ``$<$" relation.
\end{proof}
The following lemma proves the `one-to-one' 
relation between consistent cuts of event 
based and state based models of a concurrent computation.  
\begin{lemma}
\label{lem:bij}
Let $\hat{E} = (E, \ra, \pi)$ and $\hat{S} = (S, \mlt, \tau)$ be event and state based models of a concurrent  
computation.
Then there is a bijection between consistent cuts of $\hat{E}$ and $\hat{S}$. 
\end{lemma}
\begin{proof} In Appendix~\ref{sec:app_proof}. \end{proof}
Let us now study the properties of the posets
that model concurrent computations using states.
%
\section{Characteristics of State Based Models of 
Synchronous Concurrent
Computations}
\label{sec:synch}
The event based model of Defn.~\ref{def:event_model}  
accepts chain partitions that allow `shared' 
events, which in turn allows modeling synchronous
executions. We will show that 
posets that model such synchronous
concurrent computations must be width-extensible.
 We  start by showing that $\hat{S}=(S,<)$ constructed from any $(E, \ra, \pi)$ by applying 
the $ES$ transform
is width-extensible.
First, we define the three properties $\omega_1, 
 \omega_2$, and $\omega_3$ of $\hat{S}$.\\
For $1\leq i,j,k \leq n$, $\hat{S} = (S, <, \tau)$:
\begin{itemize}\itemsep0pt   
\item $(\omega_1)$
 $\forall i,j$: $[i, 0]$ $||$ $[j, 0]$. All initial states are concurrent.  
\item $(\omega_2)$  
$\forall i,j$: $[i, n_i]$ $||$ $[j, n_j]$. All final states are concurrent.  
\item $(\omega_3)$  
$\forall i,j,k$, such that for $i\neq j \wedge j\neq k$: $[i,s] < [j,t] \wedge [j,t-1] < [k,u]$ $\Ra [i,s] < [k,u]$.
\end{itemize}
We now prove that these properties are observed in $\hat{S}$.
\begin{lemma}
\label{lem:state}
$\hat{S}$ satisfies  $\omega_1$, $\omega_2$, and $\omega_3$. 
\end{lemma}
\begin{proof}
 $\omega_1$ follows immediately from the construction of $\hat{S}$ because there is no state $[i,s]$ such that $[i,s] < [i,0]$ for any $i$. That is, on any state chain $S_i$ there does not exist a state that is a precursor to the
 initial state of $S_i$. Hence, all the initial states must be concurrent.
 
  Similarly, $\omega_2$ follows when applied to the last states of $S_i$ in a similar manner because there is no state on any state chain $S_i$ that is a successor of the final state of $S_i$.
  
 $\omega_3$:  $[i,s] < [j,t] \wedge [j,t-1] < [k,u]$. Using the construction rules, we can
 infer that $(i,s+1) \ra (j,t) \wedge (j,t) \ra (k,u)$ in $E$. Which by transitivity means
 $(i,s+1) \ra (k,u)$. Hence, $[i,s] < [k,u]$ in $S$.
\end{proof}
The first condition, $\omega_1$, ensures that all $n$ initial states are pairwise concurrent. This is a valid requirement as 
all the processes would start in some default (individual) 
state, and at the start of the computation these states 
would not have any dependency amongst them. 
The second condition, given by $\omega_2$,  ensures that all $n$ final states are pairwise concurrent. This is also a valid requirement because 
irrespective of the events/commands executed, all the
$n$ processes end up in some individual final state 
at the end of the computation.
Hence, 
when the computation is finished all the final states 
would not have any dependency amongst them, and thus 
be concurrent to each other. \\
The third condition, $\omega_3$, guarantees that causal dependency 
between events under the event based model translates 
to causal dependency between corresponding states  
under the state based model. Note that the labels 
of states in the dependency relation are 
different from those of events. Suppose that for 
two events 
$e$ and $f$, we have $e \ra f$ under the 
event based model, $\hat{E}$. Then $\omega_3$
translates that dependency from $\hat{E}$ to $\hat{S}$
such that the state preceding the execution 
of $e$ is guaranteed to have existed \underline{before} 
the state that is generated \underline{after} the
execution of $f$. 

We now show that any state based model that is generated 
by applying the $ES$ transform on an event based model 
is a valid state based model. To be a valid 
state based model, it is sufficient that the generated 
poset be width-extensible. 

\begin{thm}
\label{t:subcut}
  Let $\hat{S} = (S,<,\tau)$ be a state based model for some concurrent computation. If $\hat{S}$
satisfies $\omega_1$,~$\omega_2$ and $\omega_3$,  
  then the poset $(S, <)$ is width-extensible.
\end{thm}
\begin{proof}
   We show that any antichain $A \subset S$ can be extended to a width-antichain.
  It is sufficient to show that when $|A| < n$, there exists an antichain $A \subset B$ such
  that  $|B| = |A|+1$.  Consider any process
  $C_i$ that does not contribute a state to $A$.  We will show that there
  exists a state in $S_i$ that is concurrent with all states in $A$.
  Let $s$ and $s'$ be two distinct states in $A$. 
  
  We first claim that for any state $s$ and any process $C_i$, there exists a nonempty
  sequence of consecutive states 
called the ``{\em \inx{interval}} concurrent to $s$ on $C_i$'' and 
denoted by $I_i(s)$ such that:
  \begin{enumerate}
  \item $I_i(s) \subseteq S_i$  --- i.e., the
    interval consists of only states from process $C_i$, and
  \item $\forall t \in I_i(s) : t~||~s$ --- i.e., all states in
    the interval are concurrent with $s$.
  \end{enumerate}

For a state $v \in S_i$, let $index(v)$ denote the index of state $v$ on $S_i$. Thus $0 \leq index(v) \leq n_i$. 
Define $I_i(s).lo
  = \min \{ v \: | \: v \in S_i \mya v \not\mlt s\}$.  This is well-defined
  since $[i, n_i] \not\mlt s$ due to $\omega_2$.  
Similarly, on account
  of $\omega_1$, we can define $I_i(s).hi = \max \{ v \: | \: v \in S_i \mya s
 \not\mlt v\}$.  
  We show that $I_i(s).lo \leq I_i(s).hi$ by the following case
  analysis.\\ 
  {\bf Case 1}: There exists $v:I_i(s).hi < v < I_i(s).lo$.\\
  Since $v <
  I_i(s).lo$ implies $v < s$ and $I_i(s).hi < v$ implies $s < v$,
  we get a contradiction ($v < s < v$).\\ 
  \hfill\\
  {\bf Case 2}: $index(I_i(s).hi) + 1 =  index(I_i(s).lo)$.\\
  Let $I_i(s).lo$ be the $r^{th}$ state on $S_i$, i.e., $I_i(s).lo = [i,r]$.
Then, $I_i(s).hi = [i,r-1]$. Let $s$ correspond to state $[j,t]$.
From the definition of $I_i(s).lo$, $[i,r-1] < [j,t]$.
From the definition of $I_i(s).hi$, $[j,t] < [i,r]$.
We now have, $[j,t] < [i,r]$ and $[i,r-1] < [j,t]$. From
$\omega_3$, we get $[j,t] < [j,t]$ which contradicts irreflexivity of $<$.

  From the above discussion it follows that $I_i(s).lo \leq
  I_i(s).hi$. Furthermore, for any state $t$ such that $I_i(s).lo \leq t
  \leq I_i(s).hi$, $t \not\mlt s$ and $s \not\mlt t$ holds.
Now that our claim holds, we know that $I_i(s)$ and $I_i(s')$
are both non-empty.
We show that $I_i(s)
  \cap I_i(s') \neq \emptyset$.  If not, without loss of generality assume that
  $I_i(s).hi \mlt I_i(s').lo$. Now there are two possible cases. \\
  \hfill\\
  {\bf Case 1}: $index(I_i(s).hi) + 1 = index(I_i(s').lo)$.\\
  Let $I_i(s).hi$ be $r^{th}$ state on $S_i$, i.e., $I_i(s).hi = [i,r]$.
  Then, $I_i(s').lo = [i, r+1]$. Suppose that $s = [j,u]$ and $s'=[k,v]$.
  From the definition of $I_i(s).hi$ we get that $[j,u] < [i,r+1]$.
  From the definition of $I_i(s').lo$ we get that $[i,r] < [k,v]$.
  Hence, from $\omega_3$, we get that $[j,u] < [k,v]$ --- contradicting that $s$ and $s'$ are concurrent.\\
\hfill\\
  {\bf Case 2}: There exists $v:I_i(s).hi < v < I_i(s').lo$.\\
  This implies that $s <
  v$ (because $I_i(s).hi$ precedes $v$) and $v < s'$ (because $v$
  precedes $I_i(s').lo$). Thus $s < s'$, a contradiction with
  $A$ being an antichain.  Therefore, $I_i(s) \cap I_i(s') \neq \emptyset$.  

  Because any interval $I_i(s)$ is a total order, it follows that: \[
  \bigcap_{s \in A} I_i(s) \neq \emptyset \]  We now choose any state in
  $\bigcap_{s \in A} I_i(s)$ to extend $A$.
\end{proof}

\remove{
\begin{proof}

   We show that any antichain $A \subset S$ can be extended to a width-antichain.
  It is sufficient to show that when $|A| < n$, there exists an antichain $A \subset B$ such
  that  $|B| = |A|+1$.  Consider any process
  $C_i$ that does not contribute a state to $A$.  We will show that there
  exists a state in $S_i$ that is concurrent with all states in $A$.
  Let $s$ and $s'$ be two distinct states in $A$. 
  
  We first claim that for any state $s$ and any process $C_i$, there exists a nonempty
  sequence of consecutive states 
called the ``~{\em \inx{interval}} concurrent to $s$ on $C_i$'' and 
denoted by $I_i(s)$ such that:
  \begin{enumerate}
  \item $I_i(s) \subseteq S_i$  --- i.e., the
    interval consists of only states from process $C_i$, and
  \item $\forall t \in I_i(s) : t || s$ --- i.e., all states in
    the interval are concurrent with $s$.
  \end{enumerate}

For a state $v \in S_i$, let $index(v)$ denote the index of state $v$ on $S_i$. Thus $0 \leq index(v) \leq n_i$. 
Define $I_i(s).lo
  = \min \{ v \: | \: v \in S_i \mya v \not\mlt s\}$.  This is well-defined
  since $[i, n_i] \not\mlt s$ due to $\omega_2$.  
Similarly, on account
  of $\omega_1$, we can define $I_i(s).hi = \max \{ v \: | \: v \in S_i \mya s
 \not\mlt v\}$.  

  We show that $I_i(s).lo \leq I_i(s).hi$ by the following case
  analysis.\\ 
  {\em Case 1}: There exists $v:I_i(s).hi < v < I_i(s).lo$.\\
  Since $v <
  I_i(s).lo$ implies $v < s$ and $I_i(s).hi < v$ implies $s < v$,
  we get a contradiction ($v < s < v$).\\ 
  {\em Case 2}: $index(I_i(s).hi) + 1 =  index(I_i(s).lo$.\\
  Let $I_i(s).lo$ be $r^{th}$ state on $S_i$, i.e., $I_i(s).lo = [i,r]$.
Then, $I_i(s).hi = [i,r-1]$. Let $s$ correspond to state $[j,t]$.
From the definition of $I_i(s).lo$, $[i,r-1] < [j,t]$.
From the definition of $I_i(s).hi$, $[j,t] < [i,r]$.
We now have, $[j,t] < [i,r]$ and $[i,r-1] < [j,t]$. From
$\omega_3$, we get $[j,t] < [j,t]$ which contradicts irreflexivity of $<$.

  From the above discussion it follows that $I_i(s).lo \leq
  I_i(s).hi$. Furthermore, for any state $t$ such that $I_i(s).lo \leq t
  \leq I_i(s).hi$, $t \not\mlt s$ and $s \not\mlt t$ holds.

From the above claim, we know that $I_i(s)$ and $I_i(s')$
are both non-empty.
We show that $I_i(s)
  \cap I_i(s') \neq \emptyset$.  If not, without loss of generality assume that
  $I_i(s).hi \mlt I_i(s').lo$.\\ 
  {\em Case 1}: $index(I_i(s).hi) + 1 = index(I_i(s').lo)$.\\
  Let $I_i(s).hi$ be $r^{th}$ state on $S_i$, i.e., $I_i(s).hi = [i,r]$.
  Then, $I_i(s').lo = [i, r+1]$. Suppose that $s = [j,u]$ and $s'=[k,v]$.
  From the definition of $I_i(s).hi$ we get that $[j,u] < [i,r+1]$.
  From the definition of $I_i(s').lo$ we get that $[i,r] < [k,v]$.
  Hence, from $\omega_3$, we get that $[j,u] < [k,v]$ contradicting that $s$ and $s'$ are concurrent.

  {\em Case 2}: There exists $v:I_i(s).hi < v < I_i(s').lo$.\\
  This implies that $s <
  v$ (because $I_i(s).hi$ precedes $v$) and $v < s'$ (because $v$
  precedes $I_i(s').lo$). Thus $s < s'$, a contradiction with
  $(A)$ being an antichain.  Therefore, $I_i(s) \cap I_i(s') \neq \emptyset$.  

  Because any interval $I_i(s)$ is a total order, it follows that: \[
  \bigcap_{s \in A} I_i(s) \neq \emptyset \]  We now choose any state in
  $\bigcap_{s \in A} I_i(s)$ to extend $A$.

\end{proof}
  }

\newcommand{\siminusone}{$s_{i-1}$}
\newcommand{\si}{$s_i$}
\newcommand{\tjminusone}{$t_{j-1}$}
\newcommand{\tj}{$t_j$}
\newcommand{\strel}{(\siminusone, \si) $\ra$ (\tjminusone, \tj)}
\newcommand{\surel}{(\siminusone, \si) $\ra$ $(u_{k-1}, u_k)$}
  
 We have established that every 
 poset that provides the three conditions
$\omega_1, \omega_2$, and $\omega_3$ is width-extensible. 
We now show the converse --- every
width-extensible poset guarantees these three conditions. 

\begin{thm}
\label{thm:suff}
Let  $(S, <)$ be a width-extensible poset. Consider any chain-partition $\tau$ of $(S, <)$.
Then, $\hat{S}=(S, <, \tau)$ satisfies $\omega_1$, $\omega_2$ and $\omega_3$.
\end{thm}
\begin{proof}
We show the contrapositive. If $\omega_1$ is violated, then there exists an initial state 
$t$ such that there exists a state $s$ different from $t$ which is less than  $t$.
Then, $s$ is less than all states in the process containing $t$. Therefore, the
antichain $\{t\}$ cannot be extended to a width-antichain. The proof for $\omega_2$ is dual.

If $\omega_3$ is violated, then there exist $[i,s]$, $[j,t]$ and $[k,u]$, where $i \neq j \wedge j \neq k$, such that
$[i,s] < [j,t]$ and $[j,t-1] < [k,u]$ but $[i,s] \not < [k,u]$.
We now do a case analysis on the relationship between 
$[i,s]$ and $[k,u]$.\\ 
\hfill\\
{\bf Case 1}: $[k,u] < [i,s]$. (Illustrated
in Fig.~\ref{fig:case12}, Appendix~\ref{sec:app_ills}). 
In this case we claim that there is no width-antichain
that contains $[i,s]$. Since $[i,s] < [j,t]$, for any state $w$ on process $C_j$ that is concurrent
with $[i,s]$, we get $w \leq [j,t-1]$. Since $[j,t-1] < [k,u]$ none of the states on process $C_k$ greater than
$[k,u]$ are eligible to be in the width-antichain with $w$. Furthermore, all states
less than or equal to $[k,u]$ are ineligible because $[k,u] < [i,s]$. \\
\hfill\\
{\bf Case 2}: $[k,u]$ is incomparable with $[i,s]$. In this case we claim that there is no width-antichain
 that includes both $[k,u]$ and $[i,s]$. No state greater than or equal to 
$[j,t]$ can be included from $C_j$ because $[i,s] < [j,t]$. No state less than or equal to
$[j,t-1]$ can be included from $C_j$ because $[j,t-1] < [k,u]$.\\
Note that $\omega_3$ only requires $i \neq j \wedge j \neq k$. It is possible that $i = k$; the proof still holds.
\end{proof}

\remove{
\begin{proof}
Sketch here.
We show the contrapositive. If $\omega_1$ is violated, then there exists an initial state 
$t$ such that there exists a state $s$ different from $t$ which is less than  $t$.
Then, $s$ is less than all states in the process containing $t$. Therefore, the
antichain $\{t\}$ cannot be extended to a width-antichain. The proof for $\omega_2$ is dual.
If $\omega_3$ is violated, then there exists $[i,s]$, $[j,t]$ and $[k,u]$, where $i \neq j \wedge j \neq k$, such that
$[i,s] < [j,t]$ and $[j,t-1] < [k,u]$ but $[i,s] \not < [k,u]$.
Note that $\omega_3$ only requires $i \neq j \wedge j \neq k$. It is possible that $i = k$; the proof still holds. 
We now do a case analysis on the relationship between 
$[i,s]$ and $[k,u]$.\\ 
Case 1: $[k,u] < [i,s]$. Depicted in Fig.~\ref{fig:case12}.
In this case we claim that there is no width-antichain
that contains $[i,s]$. Since $[i,s] < [j,t]$, for any state $w$ on process $C_j$ that is concurrent
with $[i,s]$, we get $w \leq [j,t-1]$. Since $[j,t-1] < [k,u]$ none of the states on process $C_k$ greater than
$[k,u]$ are eligible to be in the width-antichain with $w$. Furthermore, all states
less than or equal to $[k,u]$ are not eligible because $[k,u] < [i,s]$. \\
Case 2: $[k,u]$ is incomparable with $[i,s]$. In this case we claim that there is no width-antichain
 that includes both $[k,u]$ and $[i,s]$. No state greater than or equal to 
$[j,t]$ can be included from $C_j$ because $[i,s] < [j,t]$. No state less than or equal to
$[j,t-1]$ can be included from $C_j$ because $[j,t-1] < [k,u]$.
\end{proof}
}

With Theorems~\ref{t:subcut} and~\ref{thm:suff}, we have established that 
conditions $\omega_1$, $\omega_2$ and $\omega_3$
are necessary and sufficient for a poset to 
be width-extensible. 
We now show  that width-extensibility 
is a sufficient condition for modeling a  
concurrent computation under the 
state based model. 
First, we outline how to generate an event based 
model of a concurrent computation from $(S, <)$. Let $\tau$ be any chain partition of $(S, <)$.  
We construct an event based model $(E', \ra, \pi')$ of a concurrent computation by applying the $SE$ transform
(a reverse transform to $ES$) whose steps 
 are shown in Algorithm~\ref{alg:se}.
\begin{algorithm}[!ht]
\remove{
\begin{algorithmic}[1]
\Require Event Based Model $\hat{E} = (E,\ra,\pi)$
\Ensure State Based Model $\hat{S} = (S,<,\tau)$
\State $S_i \gets \{\}$
\For{$i=1$ to $n$}
 \For{$k=0$ to $n_i$}
   \State Add $[i,k]$ to $S_i$
 \EndFor
 \For{$k=0$ to $n_i-1$}
   \State Define $[i,k] < [i,k+1]$ in $S_i$
   \LineComment{$S_i$ is now $(|E_i| + 1)$-element chain}
 \EndFor
\EndFor
 
\State $\hat{S} \gets \bigcup_{i=1}^{n}  S_i$
 \For{$i=1$ to $n$}
  \For{$j=1$ to $n ~\wedge~ j \neq i$}
  \If{$(i,r+1) \ra (j,s)$ in $E$}
   \State Define $[i,r] < [j,s]$ in $\hat{S}$
  \EndIf
\EndFor
\EndFor
\end{algorithmic}
\caption{Algorithm~$1$: $ES$ Transform}
\label{alg:alg-1}
} 
\begin{algorithmic}[1]
\Require State Based Model $\hat{S} = (S,<,\tau)$
\Ensure Event Based Model $\hat{E} = (E',\ra,\pi)$

\State $E'_i \gets \{\}$
\For{$i=1$ to $n$}
 \For{$k=1$ to $n_i$}
   \State Add $(i,k)$ to $E'_i$
 \EndFor
 \For{$k=0$ to $n_i-1$}
   \State Define $(i,k) \ra (i,k+1)$ in $E'_i$
   \LineComment{$E'_i$ is now $(|S_i| - 1)$-element chain}
 \EndFor
\EndFor
 
\State $E'_{temp} \gets \bigcup_{i=1}^{n}  E'_i$
\For{$i=1$ to $n$}
  \For{$j=1$ to $n ~\wedge~ j \neq i$}
  \If{$[i,r-1] < [j,s]$ in $S$}
   \State Define $(i,r) \ra (j,s)$ in $E'_{temp}$
  \EndIf
\EndFor
\EndFor

\State $E' \gets E'_{temp}$ 

\ForAll{$C_s$ in SCC-Decomposition of $E'_temp$} 
  \If{each node is $C_s$ lies on diff. chains} 
      \State Replace $C_s$ with one element $e$ in $E'$
      \State Assign all labels of nodes in $C_s$ to $e$
  \Else
      \State Report $S$ as {\bf not} width-extensible
  \EndIf
\EndFor
\end{algorithmic}
\caption{$SE$ (State to Event) Transform}
\label{alg:se}
\end{algorithm}

\remove{

\begin{algorithm}[]
\DontPrintSemicolon
\SetKwFunction{Diffchain}{AllOnDifferentChain}
\SetCommentSty{small}
\SetProcNameSty{textbf}
\KwIn{State Based Model $\hat{S} = (S,<,\pi)$}
\KwOut{Event Based Model $\hat{E} = (E,\ra,\pi')$ OR 
 Report $\hat{S}$ invalid}
\BlankLine
\For{$i=1 \ldots n$}{
 $E'_i = \{\}$\;
 \For{$k=1 \ldots n_i$}{
   Add $(i,k)$ to $E'_i$\; 
 }
 \For{$k=1 \ldots n_i-1$}{
   Define $(i,k) \ra (i,k+1)$ in $E'_i$\; 
 }
}
\BlankLine
\tcc*[l]{
 $E'_i$ is now $(|S_i| - 1)$-element chain }
$E'_{temp} = \cup  (E'_i ~ | ~ i = 1, 2 \ldots n)$\;
\BlankLine
\For{$i=1 \ldots n$}{
 \For{$j=1 \ldots n$ ~\&\&~ $j \neq i$}{
  \If{$[i,r-1] < [j,s]$ in $S$}{
   Define $(i,r) \ra (j,s)$ in $E'_{temp}$\; 
   }
}
}
\BlankLine
$E' = E'_{temp}$\;
\caption{$SE$ Transform}
\label{alg:se}
\end{algorithm}
\begin{algorithm}[]
\DontPrintSemicolon
\SetKwFunction{Diffchain}{AllOnDifferentChain}
\SetCommentSty{small}
\SetProcNameSty{textbf}
\ForEach{SCC $C_s$ in SCC Decomposition of $E'_{temp}$}{
   {\bf if}~each node in $C_s$ is on a different chain\;
\Indp
      Replace $C_s$ with one element $e$ in $E'$\;
      Assign all labels of nodes in $C_s$ to $e$ in $E'$ \;
      \Indm
   {\bf else}\;
\Indp
      Report $S$ as {\bf not} width-extensible\;
      \Indm
   
}
\remove{
\BlankLine
{\bf routine} \Diffchain$(SCC~~ C_s)$\{\;
\Indp
\remove{
   $k$ = \# of nodes in $C_s$\; 
   \ForEach{$e \in C_s$}{
      \If{$\pi(e) \not\in \mathcal{A}$}{
         $\mathcal{A} = \mathcal{A}  + \pi(e)$\;
      }
   }
}
   {\bf if}~each node in $C_s$ is on a different chain\;
\Indp
      return {\bf true}\;
\Indm
   {\bf else}\;
\Indp
      return {\bf false}\;
\Indm
}
\Indm
\}\;   
\caption{Auxilliary procedures for detecting cycles}
\label{alg:se}
}

In the algorithm, lines $1-11$ perform a reversal of steps of $ES$ 
transform. Lines $13-18$ try to collapse events 
that are `shared' between processes by performing a 
strongly connected component (SCC) decomposition, 
and using the SCCs for identifying shared events. 
If an SCC has events from the same process, then 
that results in a same process cycle --- an 
invalid event based computation. If we represent $\hat{S}$ as a directed graph with $m = |S|$ vertices 
and $d$ directed edges, then the complexity of $SE$ transform is $\mathcal{O}(m + d)$, i.e. linear in size of the graph. 

See Appendix~\ref{sec:app_ills} for some illustrations 
of $SE$ transform's application to examples discussed in 
this paper. 
\remove{
Fig.~\ref{fig:etemp} shows the $E'_{temp}$ (and 
not the final $E'$) generated
during the execution when $SE$ transform is applied 
to $\hat{S}$ given by Fig.~\ref{fig:barrier-koh-state}. After the SCC decomposition based `collapsing' on this $E'_{temp}$, the generated $E'$ is same 
as \figref{fig:barrier-koh-event}.  
Recall that we claimed 
invalidity of a state based model poset shown in
Fig.~\ref{fig:invalidB} 
claiming that such a state model would 
cause cycles when converted to an event based model. 
Let us assign state labels to the 
states shown in that figure: 
$a=[1,0], b=[1,1], c=[1,2], d=[2,0], 
e=[2,1]$.
Now apply the
$SE$ transform of Alg.~\ref{alg:se}
to this poset on states. The resulting $(E,\ra)$ 
would be the one shown in ~\figref{fig:eventcycle}. 
Such a cycle can not
exist in a valid event based model.
\tikzstyle{blackCirc}=[circle,draw=black,fill=black,thick,inner sep=0pt,minimum size=2mm]
\begin{figure}
\centering
\begin{subfigure}[b]{0.45\textwidth}
\centering
\begin{tikzpicture}[scale=0.65]
\node at ( 0,2) [blackCirc] [label={above:$(1,1)$}] (a) {};
\node at (2.5,2) [blackCirc] [label={above:$(1,2)$}] (b) {};
\node at (2.5,0) [blackCirc] [label={below:$(2,2)$}] (b1) {};
\node at (5.5,2) [blackCirc] [label={above:$(1,3)$}] (c) {};
\node at (0,0) [blackCirc] [label={below:$(2,1)$}] (e) {};
\node at (6.2,0) [blackCirc] [label={below:$(2,3)$}] (g) {};
\node at (7.2, 2) [] (end1) {};
\node at (7.2, 0) [] (end2) {};
\draw [thick,->] (a.east) -- (b.west);
\draw [thick,->] (b.east) -- (c.west);
\draw [thick,->] (b.east) -- (g.west);
\draw [thick,->] (b1.east) -- (g.west);
\draw [thick,->] (e.east) -- (b1.west);
 \draw [thick,->]  (b) edge[out=260,in=110,->] (b1);
 \draw  [thick,->] (b1) edge[out=60,in=290,->] (b);
\end{tikzpicture}
\caption{$E'_{temp}$ for $\hat{S}$ of ~\figref{fig:barrier-run-state}}
\label{fig:etemp}
\end{subfigure}
\begin{subfigure}[b]{0.45\textwidth}
\centering
\begin{tikzpicture}[scale=0.65]
\node at (0,2) [blackCirc] [label={above:$(1,1)$}] (a) {};
\node at (3,2) [blackCirc] [label={above:$(1,2)$}] (b) {};
\node at (1.5,0) [blackCirc] [label={below:$(2,1)$}] (e) {};
\draw [thick,->] (a.east) -- (b.west);
\draw [thick,->] (b) -- (e);
\draw [thick,->] (e) -- (a);
\end{tikzpicture}
\caption{Event model generated from states of Fig.~\ref{fig:invalidB}}
\label{fig:eventcycle}
\end{subfigure}
\caption{$SE$ transform applied to earlier examples}
\end{figure}
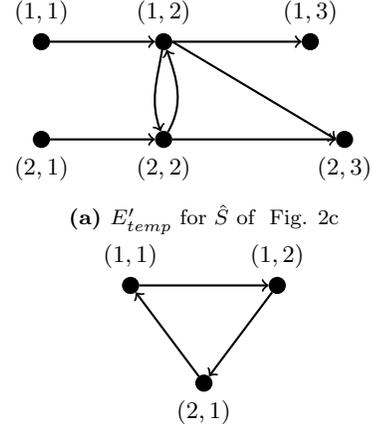
}

The next theorem  shows that width-extensibility is sufficient 
for modeling concurrent computations under the state
based model. 
\begin{thm}
\label{thm:se_valid}
Let $(S, <)$ be any width-extensible poset. Then, there exists
a concurrent computation for which it is the state based model.
\end{thm}
\begin{proof}
We show that there exists a concurrent computation in the event based model such that
when we convert that event based computation to state based model, we get
the poset $(S,<)$.

We first create a width chain partition $\tau$ of $(S, <)$ to get $(S,<,\tau)$.
We then generate an event based model $\hat{E'} = (E', \ra, \pi')$ from $(S,<,\tau)$ using $SE$ transform.
It can be easily verified that applying the $ES$ transform to $(E', \ra, \pi')$ leads to
 $(S, <, \tau)$.
It suffices to show that $(E', \ra)$ is a partial order. \\
{\em Irreflexivity}: Assume, $(i,r) \ra (i,r)$ in $E'(\pi')$. This would require
$r < r$ in $\hat{S}$ --- a contradiction. \\
{\em Transitivity}: Consider $(i,r) \ra (j,s) \wedge  (j,s) \ra (k,t)$, in $E'$. First, let us look at the 
case where $i \neq j \wedge j \neq k$. 
$(i,r)\ra(j,s)$ in the event based model is possible only if
$[i, r-1] < [j, s]$ in $\hat{S}$. Similarly, we also get $[j,s-1] < [k,t]$. Hence:
\[
[i, r-1] < [j, s] \wedge [j,s-1] < [k,t]
\]
By using $\omega_3$ on $\hat{S}$  we get
$[i,r-1] < [k,t]$ in $\hat{S}$
$\equiv (i,r) \ra (k,t)$ in $E'$.\\
When $i = j = k$, the transitivity of states the same chain 
is trivial. Now let us consider the case when $i = j \wedge j \neq k$. Then, $(i,r) \ra (j,s) \wedge  (j,s) \ra (k,t)$ in
$E'$
requires $r < s$, as $i = j$, and $[j,s-1] < [k,t]$ in $\hat{S}$. Observe that  $i = j$ and $r < s$ means that $r-1$, $s-1$, $s$ form a totally ordered set, such that
$r-1 \leq s -1$. Hence, we get $[i, r-1] \leq [j,s-1] \wedge [j,s-1] < [k,t]$. By transitivity of $<$ in $\hat{S}$, this leads to $[i,r-1] < [k,t]$ which is the
desired condition for $(i,r) \ra (k,t)$ in $E'$. The proof for the case of $i \neq j, j = k$ is similar.
Finally, consider the case when $i = k, i \neq j \wedge r = t$. In such a case, the original condition in the $E'$
becomes $(i,r)\ra(j,s)\wedge(j,s)\ra(i,r)$. Given that 
we have $i\neq j$, the condition is only 
possible if $(i,r)$ and $(j,s)$ represent the same 
shared event --- shared between processes/chains $i$ and $j$. 
Now that $(i,r)$ and $(j,s)$ represent the same shared
event, the requirement of transitivity on this 
event is trivially held.  
\end{proof}

\remove{
\begin{proof}
Sketch here.
We show a concurrent computation in the event based model such that
when we convert that event based computation to state based model, we get
the poset $(S,<)$.
We first create a width chain partition $\tau$ of $(S, <)$ to get $(S,<,\tau)$.

\remove{
define the elements of $E'_i$ as:
$E'_i :(i, 1)\ra \ldots \ra(i, n_i -1)\ra(i, n_i)$.
Create a happened-before relation between
these elements (events) of $E'$ as follows:
$(i, r) \ra (j, s)$ in $E' (i, j = 1, 2,\ldots n;
1\leq r\leq n_i , 1 \leq s \leq n_j)$ {\bf if}:\\
$r < s, $ and $i = j$ in $\hat{S}$\\
or\\
$[i, r-1]$ and $[j, s]$ \\are defined in $\hat{S}$ and $[i, r-1] < [j, s]$ in $\hat{S}$, if $i \neq j$.
}
We then generate an event based model $\hat{E'} = (E', \ra, \pi')$ from $(S,<,\tau)$ using $SE$ transformation.
It can be easily verified that applying the $ES$ transformation of Alg.~$1$ (\figref{alg:alg-1})
 to $(E', \ra, \pi')$ leads to
 $(S,<,\tau)$.
It suffices to show that $(E', \ra)$ is a partial order. \\
{\em Irreflexivity}: Assume, $(i,r) \ra (i,r)$ in $E'(\pi')$. This would require
$r < r$ in $\hat{S}$ - a contradiction. \\
{\em Transitivity}: Consider  $(i,r) \ra (j,s) \wedge  (j,s) \ra (k,t)$ where $i \neq j \wedge j \neq k$. The first relation in the event based model is possible only if
$[i, r-1] < [j, s]$ in $\hat{S}$. Similarly, we also get $[j,s-1] < [k,t]$. Hence:
\[
[i, r-1] < [j, s] \wedge [j,s-1] < [k,t]
\]
By using $\omega_3$ on $\hat{S}$  we get
$
[i,r-1] < [k,t]
\equiv (i,r) \ra (k,t)
$
in $E'$.\\
When $i = j = k$, the transitivity of states on the one chain can be shown trivially. Now let us consider the case when $i = j \wedge j \neq k$. Then, $(i,r) \ra (j,s) \wedge  (j,s) \ra (k,t)$ in
$E'$
requires $r < s$, as $i = j$, and $[j,s-1] < [k,t]$ in $\hat{S}$. Observe that  $i = j$ and $r < s$ means that $r-1$, $s-1$, $s$ form a totally ordered set, such that
$r-1 \leq s -1$. Hence, we get $[i, r-1] \leq [j,s-1] \wedge [j,s-1] < [k,t]$. By transitivity of $<$ in $\hat{S}$, this leads to $[i,r-1] < [k,t]$ which is the
desired condition for $(i,r) \ra (k,t)$ in $E'$. The proof for the case of $i \neq j, j = k$ is similar.\\
\remove{
Finally, consider the case when $i = k, i \neq j \wedge r = t$. In this case, the left hand side of  $(i,r) \ra (j,s) \wedge  (j,s) \ra (k,t)$ is equivalent to
$[i,r-1] < [j,s] \wedge [j,s-1] < [i,r]$ as $i = k, r =t$. However note that $(S4)$ prohibits this case. Hence, the left hand side is false; and thus the condition is trivially true.
}

\end{proof}
}

The following lemma combines the results established earlier to show that $ES$ and $SE$ 
transforms are inverse functions of each other. 
\begin{lemma}
Let $\hat{E}=(E,\ra,\pi)$ be an event based model for some computation and let $\hat{S}$ 
be the result of applying $ES$ transform to $\hat{E}$. Then, applying $SE$ transform on $\hat{S}$
results in $\hat{E}$. 
\label{lem:combine1}
\end{lemma}
\begin{proof}
Follows directly from lemmas~\ref{lem:poset},~\ref{lem:bij}, and~\ref{lem:state} combined 
with theorems~\ref{t:subcut},~\ref{thm:suff}, and~\ref{thm:se_valid}. 
\end{proof}

Thus, we have established that $\omega_1$, $\omega_2$, and $\omega_3$ properties provide a complete 
characterization of a state based model for a 
concurrent computation. In the next section, 
we discuss asynchronous  
computations, and show that their state based
models are a special case 
of models of concurrent computations formalized
in this section. 

\section{Characteristics of State based Models
of Asynchronous\\ Concurrent Computations}
\label{sec:adc}
Asynchronous concurrent
computations, which are common in 
distributed systems, 
are a special type the concurrent computations
that cannot 
have any `shared' events. Shared events are only 
possible when the communication between 
processes is synchronous. Thus, 
the event based model of asynchronous 
computations is defined based on 
a chain partition $\pi$ in which 
all chains are disjoint.  
The event based model of asynchronous concurrent computations(we use the short-form notation ASC from here on) is given by the following definition:  
\begin{definition}[Event based model of ASC]
\label{def:adc_event}
An event based model of an ASC on $n$ processes is 
 is a tuple $(E, \ra, \pi)$ where
$E$ is the set of events, $\ra$ is the happened-before  relation on $E$,
and $\pi$ is a map from $E$ to $\{1..n\}$ such that 
for all distinct events $e,f \in E:~ \pi(e) = \pi(f) \Ra (e \ra f) \vee (f \ra e)$. 
\end{definition}
Thus, $\pi$ partitions $E$ such that every block of the partition is  totally
ordered under $\ra$.

Such an event based model, with no `shared' events, leads to a state based model 
that satisfies stronger properties than those 
satisfied
by the state based model of the previous section.
Intuitively, 
given that the communication 
between processes is asynchronous, no two processes 
can make a `jump' together from their individual 
states to next states as if there was a `shared' 
execution. Hence, the poset $(S,<)$ exhibits
a property that we call `interleaving-consistency'.

\begin{definition}[Interleaving-consistent Poset]
A poset $(X, <)$ is interleaving-consistent if for every
width-antichain $W$ that is not equal to the 
biggest width-antichain, there exists a width-antichain $W' > W$ such that
$|W \cap W'|  = |W|-1$.
\end{definition}

Let $\mathcal{A}(X)$ be the set of all width-antichains 
of a poset $(X,<)$. The biggest width-antichain of $(X,<)$ is the width-antichain $A \in \mathcal{A}(X)$ such that $\nexists A' \in \mathcal{A}(X): A < A'$. 
Informally, interleaving-consistency requires that any possible cut (modeled
as a width-antichain) can be advanced on some process to reach another
possible cut. Fig.~\ref{fig:koh_transform}\subref{fig:ex_chains} shows
an ASC under the event based model, and the corresponding 
poset of the 
state based model in Fig.~\ref{fig:koh_transform}\subref{fig:sppi} is 
interleaving-consistent. In contrast, the event based 
computation in \figref{fig:barrier-koh-event} is not an ASC, and thus the resulting
state based model's poset in 
\figref{fig:barrier-koh-state} 
is not interleaving-consistent --- the 
processes make a `jump' together from states 
$[1,1], [2,1]$ to $[1,2], [2,2]$. 

ASCs are a special kind (subset) of 
concurrent computations, and thus
 a partial order modeling states of an ASC must satisfy $\omega_1$, $\omega_2$ and 
$\omega_3$. In addition, it should also be {\em interleaving-consistent}. 
Formally, {\em interleaving-consistency} of $\hat{S}$ 
 is captured by the condition $\psi$ as follows:\\
\hfill\\
\framebox{
\parbox{0.45\textwidth}{
$(\psi)$ for $1\leq i,j \leq n$,~~
if~~$i\neq j$,~~ then~~ $[i,s-1] < [j,t] \Ra \neg~([j,t-1] < [i,s]).$
}
}
\remove{
\[(\psi)~~for~~ 1\leq i,j \leq n,~~
if~~i\neq j,~~ then~~ [i,s-1] < [j,t] \Ra \neg~([j,t-1] < [i,s]).\]
}

 Thus, for an ASC, a poset $(S, <)$ that models its states
 is characterized by  $\omega_1$, $\omega_2$, 
$\omega_3$, and $\psi$. 
The state based model for ASCs is formally defined as:
\begin{definition}[State based model of ASCs]
\label{def:adc_state}
An asynchronous distributed computation on 
$n$ processes is modeled by $\hat{S}=(S, <, \pi)$, where
$S$ is the set of states, $<$ is an irreflexive partial order relation on $S$
such that $(S, <)$ is a width-extensible and interleaving-consistent poset,
and $\pi$ maps every state to a process from $\{1..n\}$ such that
for all $i \in \{1 \ldots n\}, S_i = \{s \in S~|~i \in \pi(s) \}$ is totally ordered under $<$.
\end{definition}

%

The following set of results establish the properties
of state based models of ASCs. 
\begin{lemma}
\label{lem:al4}
Suppose  $\hat{S}=(S, <, \tau)$ is obtained 
by applying $ES$ transform on an ASC's event based 
model $\hat{E}=(E, \ra, \pi)$. Then $\hat{S}$
 satisfies  $\omega_1$, $\omega_2$, $\omega_3$, and
$\psi$. 
 \end{lemma}
\begin{proof}
Since ASCs are a subset of concurrent computations, the conditions 
$\omega_1, \omega_2, \omega_3$ continue to be satisfied as shown 
in Theorem~\ref{thm:suff}. 
Suppose $(S,<)$ doesn't satisfy $\psi$ and 
thus we have $[i,s-1] < [j,t] \Ra ([j,t-1] < [i,s])$ in 
$(S,<)$. But this would require $(i,s) \ra (j,t) \wedge 
(j,t) \ra (i,s)$ in $E$, which is a contradiction. 
\end{proof}

\begin{lemma}
\label{lem:al4ic}
Let $\hat{S}=(S,<,\tau)$ be as defined in Lemma~\ref{lem:al4}. 
Then, $(S, <)$ is interleaving-consistent.
\end{lemma}
\begin{proof}
 Suppose $(S,<)$ satisfies 
 $\psi$, but is not interleaving-consistent. 
 Hence, there is some antichain $A$ of 
 $(S,<)$ that 
 is not the biggest, and still can not be extended along just one process to form another antichain 
 $A'$. Let $[i,a_i]$ be the element from chain $i$ that belongs to $A$. Our assumption 
 means that $\nexists i: A - \{[i,a_i]\} + \{[i,a_i+1]\}$ is a width-antichain.
Hence $\forall i, \exists j \neq i: [i,a_i] < [j,a_j+1]$. Given that $S$ is finite (we can not
keep on finding a `new' $j$ for every `new' $i$ we consider), we know that to satisfy this 
requirement
there must exist $k, k\neq j \wedge k\neq i$ such that 
$[j,a_j] < [i,a_i+1] \wedge [k,a_k] < [j,a_j+1] \wedge [i,a_i] < [k,a_i+1]$. See Fig.~\ref{fig:ijkal4} in Appendix~
\ref{sec:app_ills}
 for an illustration. 

From the previous lemma, we know that $(S,<)$ 
 satisfies $\omega_1, \omega_2, \omega_3$. 
Applying $\omega_3$ we get $[k,a_k] < [i,a_i+1]$. But this leads to $[i,a_i] < [k,a_i+1]
\wedge
[k,a_k] < [i,a_i+1]$ --- a contradiction with $\psi$. 
\end{proof} 

\begin{thm}
Let  $(S, <)$ be any interleaving-consistent and width-extensible poset. Consider any chain partition $\tau$ of $(S, <)$.
Then, $\hat{S}=(S, <, \tau)$ satisfies  $\omega_1$, $\omega_2$, $\omega_3$, and
$\psi$. 
\end{thm}
\begin{proof}
Width-extensibility guarantees $\omega_1$, $\omega_2$ and $\omega_3$. 
Suppose $\psi$ is not satisfied, i.e., there exist 
$[i,s] , [j,t]$, for $i\neq j$, such that $[i,s-1] < [j,t] \wedge ([j,t-1] < [i,s])$.
Let $[j,r]$ be  the largest state on $S_j$ that is
incomparable with $[i,s-1]$. Note that $r \leq t-1$. It is clear that $[j,r] < [j,t]$ because $[i,s-1] < [j,t]$.
Since $([j,t-1] < [i,s])$, we also get that $[j,r] < [i,s]$.

Let ${\cal W}$ be the set of all width-antichains that include both $[i,s-1]$ and $[j,r]$.
Let $A$ be the biggest antichain in ${\cal W}$.
We claim that there does not exist any width-antichain $A' \geq A$ such that $|A' - A| = 1$, and thus not satisfying
$\psi$ contradicts with interleaving-consistency.
If $A'$ differs from $A$ on a chain different from $i$ and $j$, then it violates that $A$ is the biggest antichain
that contains $[i,s-1]$ and $[j,r]$. Hence, to satisfy interleaving-consistency, $A'$ must differ from $A$ on either $i$ or $j$.
 Suppose $A' - A = [i,s]$ then, because $A'$ is a width-antichain, we get that $[j,r]$ is incomparable with $[i,s]$, a
contradiction. If $A' - A = [j,r+1]$, then we get that $[i,s-1]$ is incomparable with $[j,r+1]$, which contradicts 
the definition of $[j,r]$. 
\end{proof}

\remove{
\begin{proof}
Width-extensibility guarantees $\omega_1$, $\omega_2$ and $\omega_3$. 
Suppose $\psi$ is not satisfied, i.e., there exist 
$[i,s] , [j,t]$, for $i\neq j$, such that $[i,s-1] < [j,t] \wedge ([j,t-1] < [i,s])$.
Let $[j,r]$ be  the largest state on $S_j$ that is
incomparable with $[i,s-1]$. Note that $r \leq t-1$. It is clear that $[j,r] < [j,t]$ because $[i,s-1] < [j,t]$.
Since $([j,t-1] < [i,s])$, we also get that $[j,r] < [i,s]$.

Let ${\cal W}$ be the set of all width-antichains that include both $[i,s-1]$ and $[j,r]$.
Let $A$ be the biggest antichain in ${\cal W}$.
We claim that there does not exists any width-antichain $A' \geq A$ such that $|A' - A| = 1$, and thus not satisfying
$\psi$, contradicts with interleaving-consistency.
If $A'$ differs from $A$ on a chain different from $i$ and $j$, then it violates that $A$ is the biggest antichain
that contains $[i,s-1]$ and $[j,r]$. Hence, to satisfy interleaving-consistency, $A'$ must differ from $A$ on either $i$, or $j$.
 Suppose $A' - A = [i,s]$ then because $A'$ is a width-antichain -- we get that $[j,r]$ is incomparable with $[i,s]$, a
contradiction. If $A' - A = [j,r+1]$, then we get that $[i,s-1]$ is incomparable with $[j,r+1]$, which contradicts 
the definition of $[j,r]$. 
 
\end{proof}
}

\begin{thm}
\label{thm:exists}
Let $(S, <)$ be any poset that is width-extensible and interleaving-consistent. Then, there exists
an ASC for which it is the state-based model.
\end{thm}
\begin{proof}
Let $\tau$ be any chain partition of $(S, <)$.
Apply $SE$ transform on \sposet to generate
an event based model $(E',\ra)$.
It is trivial to verify that applying $ES$ transform 
 to $(E',\ra)$ leads to
 $(S, <)$.
It suffices to show that $(E',\ra)$ is a partial order. \\
{\em Irreflexivity}: can be proved using exactly the 
same argument used in Theorem~\ref{thm:se_valid}.\\
{\em Transitivity}: Except the case of $i = k, i \neq j \wedge r = t$, apply the same argument as in Theorem~\ref{thm:se_valid}.
For $i = k, i \neq j \wedge r = t$, we use a different
argument. In this case, the left hand side of  $(i,r) \ra (j,s) \wedge  (j,s) \ra (k,t)$ is equivalent to
$[i,r-1] < [j,s] \wedge [j,s-1] < [i,r]$ as $i = k, r =t$. But, $\psi$ prohibits this case --- hence the left hand side is false and the constraint holds trivially.
\end{proof}

Similar to Lemma~\ref{lem:combine1}, we can now verify the following result. 
\begin{lemma}
Let $\hat{E}=(E,\ra,\pi)$ be an event based model for some ASC and let $\hat{S}$ 
be the result of applying $ES$ transform to $\hat{E}$. Then, applying $SE$ transform on $\hat{S}$
results in $\hat{E}$. 
\end{lemma}
\begin{proof}
Follows directly from lemmas~\ref{lem:bij},~\ref{lem:al4}, and~\ref{lem:al4ic} combined 
with theorems~\ref{t:subcut}, and~\ref{thm:exists}. 
\end{proof}

\newcommand{\lwa}{L_{WA}}

\section{Applications}
\label{sec:checkpoint}
To conclude, we now discuss 
two applications of duality between state and event based models of concurrent computations.

\subsection{Predicate Detection} 
Our theory applies to detection of global predicates that depend only on the latest events in ASCs.
For example, consider a set of processes that execute three
kinds of events: internal, message send and {\em blocking receive}. The blocking
receive event blocks the process until it receives a message from some process.
It is clear that in absence of in-transit messages, and the last executed event at all processes
being a receive event, the system has a communication deadlock.
In this example, we require that the last event at each process be a blocking receive.
Even if one process is left out, that process could send messages to all other processes to
unblock them.

Recall that an ideal
$Q$  of a poset $P = (X,\leq)$ is a {\em width-ideal} if 
the set of all maximal elements in $Q$, denoted by $maximal(Q)$, is a 
width-antichain of $P$. Let $B$ be a predicate, and $G$ be a global state of a computation, then $B(G)$ denotes 
that $B$ is true on $G$. A width-predicate is defined as follows. 
\begin{definition}[Width-Predicate]
A global predicate $B$ in a distributed computation on $n$ processes is a
{\em width-predicate} if
$B(G) \Rightarrow |maximal(G)| = n$.
\end{definition}
Some examples of width-predicates are:\\
$1.$ {\em Barrier synchronization}:  ``Every process has made a call
to the method {\sf \small barrier}.''\\
$2.$
{\em Deadlock for Dining Philosophers}: ``Every philosopher has picked up
a fork''.\\
$3.$
{\em Global Availability}: ``Every process has an active session and the total
number of permits with processes is less than $k$.''

Note that $1$ and $2$ are also conjunctive predicates
and can already be detected efficiently. But even if $B$ is not stable or conjunctive, as in example $3$, we can use our theory to detect it.
Clearly, to detect a width-predicate, it is sufficient to construct or traverse the lattice
of the width-ideals.
The following result, based on \cite{Ganter10, Gar03}, gives an idea for an algorithm to construct or traverse the lattice.  
\begin{thm}
Given any finite width-extensible poset $P$, there exists an algorithm to
enumerate all its width-antichains in
$O(n^2L)$ time where $n$ is the width of the poset and $L$ is the size
of the lattice of width-antichains.
\end{thm}
\begin{proof}
We exploit the bijection between the set of all down-sets of $(E, \ra, \pi)$ and
the set of all width antichains of $(S, <, \tau)$ (Lemma 1). Given the
poset $P$, we apply the $SE$
transform (Algorithm~\ref{alg:se}) to get another poset $P'$ such that enumerating consistent cuts of $P'$ is equivalent to enumerating all
width-antichains
of $P$. We can now use algorithms in \cite{Ganter10, Gar03} on $P'$ to enumerate
all down-sets in $O(n^2L)$ time.
\end{proof}

\subsection{Better Understanding of Checkpointing}
 Checkpointing \cite{bhargava} is widely used for   
fault tolerance in distributed systems. In {\em uncoordinated}
checkpointing \cite{xu:zig}, processes take checkpoints
independently, without any group communication and coordination.
In a distributed computation with $n$ processes, 
let $L_i$ denote the sequence of local checkpoints of process $C_i$.
Note that any checkpoint $lc \in L_i$ is a local state of process $C_i$. Hence, $L_i$ 
is a state chain that is totally ordered under the ``$<$'' relation that we have used for 
comparing states in this paper.  
 It is common to assume that the initial state and the final state
in each process are checkpointed \cite{xu:zig,helary97preventing}.
Let the set of all local checkpoints be $L$, i.e., $L = \bigcup L_i.$
The set of checkpoints $L$, together with the existed-before
relation ``$<$'', forms a state based model $\hat{L} = (L,<,\tau)$,
where $\tau$ partitions $L$ into chains. A subset $G \subseteq L$
is a global checkpoint iff $\forall c,d \in G: c~\|~d$ and
$|G| = n$. Hence, a global checkpoint is equivalent to a consistent global state 
in a state based model over checkpoints of the computation. 
 A local checkpoint is `useless' if it cannot be part of any 
global checkpoint. 
Netzer et al. \cite{xu:zig} established results on useless
checkpoints using the notion of {\em zig-zag} paths.   
Wang \cite{wang:rdt} 
used a construction called
$R$-graph (or, rollback-dependency graph) to devise an algorithm for detection of useless checkpoints. Although, both
\cite{xu:zig,wang:rdt}
 have made important contributions, they do not clearly highlight the fundamental concept that checkpoints 
 are states of a distributed computation, and reasoning 
 about checkpoints is in effect reasoning over 
 the state based model of an ASC. Using the theory established in this paper, one can easily understand the  
 intuition behind constructions of {\em zig-zag} paths and $R$-graphs. 
 In short, by viewing a checkpointing computation 
 as a state based model, the interpretation of {\em zig-zag} paths, useless checkpoints, and $R$-graphs is as follows.
 \begin{itemize}\itemsep0em
 \item Absence of {\em zig-zag} paths between checkpoints (states) in a computation means that the 
 checkpoints can be part of a width-antichain (consistent cut).   
Their presence between checkpoints indicates that the 
 checkpoints cannot be part of a width-antichain.
A useless checkpoint is a state of a computation that 
 cannot belong to any width-antichain of the 
poset under the state based model. 
 \item The $R$-graph construction on a checkpoint computation essentially generates an event based model
 from the state based model that is the original computation. Hence, the algorithm to identify 
 useless checkpoints (in \cite{wang:rdt}) effectively tries to check if the model 
of the computation is legal under the  event based model when a particular checkpoint is included. Thus, it applies
the $SE$ transform (Alg.~\ref{alg:se}) on the state based model imposed by the checkpoints.  
 The $R$-graph construction and detection algorithm (by finding cycles) and the $SE$ transform
have the same computation complexity $\mathcal{O}(k + m)$, where $k$ is the number of checkpoints (states) in the computation
and $m$ is the number of messages. 
 \end{itemize}
\remove{
Thus, even though our work does not better the $R$-graph based algorithm, the concepts presented 
here provide a way of clearly understanding why the $R$-graph based algorithm works.  }

\remove{
A zig-zag path between two checkpoints $c$ and $d$
is defined as follows. There is a zig-zag path from $c$ to $d$ iff:\\
{\bf(a)} Both $c$ and $d$ are in the same process and $c \ra d$.\\
{\bf(b)} There is a sequence of messages $m_1,\ldots,m_t$ such that: \\
\h  (i) $m_1$ is sent after the checkpoint $c$.\\
\h  (ii) If $m_k$ is received by process $r$, then $m_{k+1}$ is sent by
       process $r$ in the same or a later checkpoint interval. Note that the message $m_{k+1}$
       may be sent before $m_k$.\\
\h  (iii) $m_t$ is received before the checkpoint $d$.}


\bibliographystyle{abbrv}
\bibliography{refs,nlattice}
%
\appendix
\section{Proof of Lemma 1}
\label{sec:app_proof}

\begin{flushleft}
{\bf Lemma 1}. {\em 
Let $\hat{E} = (E, \ra, \pi)$ and $\hat{E} = (S, \mlt, \tau)$ be event and state based models of a computation.
Then there is a bijection between consistent cuts of $\hat{E}$ and $\hat{S}$.} 
\end{flushleft}
\begin{proof}
Let $G$ be any consistent cut of $(E, \ra, \pi)$. We will show how to construct the corresponding
consistent cut $T$ of \sposet.
Suppose that $G$ contains at least one event from $C_i$.
Then, let $(i,k)$ be the largest event from process $C_i$. In this case, we add $[i,k]$ to $T$.
If $G$ does not contain any event from $C_i$, then we add $[i,0]$ to $T$.
Clearly, $T$ has exactly $n$ states, one from each process.  We show that the cut $T$ is also consistent.
If not, suppose $[i,s]$ and $[j,t]$ be two states in $T$ such that $[i,s] < [j,t]$. This implies that
$(i,s+1) \ra (j,t)$, under the event based model, contradicting that $G$ is consistent because $G$ contains $(j,t)$ but does not contain $(i,s+1)$.
It is also easy to verify that the mapping from the set of consistent cuts is one-to-one.

Conversely, given a consistent cut $T$ 
in the state based model, we construct a consistent 
cut in event based model in $1-1$ manner as follows. For all states $[i,k] \in T$ we include
all events $(i,k')$ such that $k ' \leq k$. Note that when $k$ equals $0$, no events from $C_i$ are included.
It can again be easily verified that whenever $T$ is a consistent 
cut in state model, $G$ is a consistent cut in event model.
\end{proof}
\section{Illustrations}
\label{sec:app_ills}
Fig.~\ref{fig:etemp} shows the $E'_{temp}$ (and 
not the final $E'$) generated
during the execution when $SE$ transform is applied 
to $\hat{S}$ given by Fig.~\ref{fig:barrier-koh-state}. After the SCC decomposition based `collapsing' on this $E'_{temp}$, the generated $E'$ is same 
as \figref{fig:barrier-koh-event}.  
Recall that we claimed 
invalidity of a state based model poset shown in
Fig.~\ref{fig:invalidB} 
claiming that such a state model would 
cause cycles when converted to an event based model. 
Let us assign state labels to the 
states shown in that figure: 
$a=[1,0], b=[1,1], c=[1,2], d=[2,0], 
e=[2,1]$.
Now apply the
$SE$ transform of Alg.~\ref{alg:se}
to this poset on states. The resulting $(E,\ra)$ 
would be the one shown in ~\figref{fig:eventcycle}. 
Such a cycle can not
exist in a valid event based model.
\tikzstyle{blackCirc}=[circle,draw=black,fill=black,thick,inner sep=0pt,minimum size=2mm]
\begin{figure}
\centering
\begin{subfigure}[b]{0.45\textwidth}
\centering
\begin{tikzpicture}[scale=0.65]
\node at ( 0,2) [blackCirc] [label={above:$(1,1)$}] (a) {};
\node at (2.5,2) [blackCirc] [label={above:$(1,2)$}] (b) {};
\node at (2.5,0) [blackCirc] [label={below:$(2,2)$}] (b1) {};
\node at (5.5,2) [blackCirc] [label={above:$(1,3)$}] (c) {};
\node at (0,0) [blackCirc] [label={below:$(2,1)$}] (e) {};
\node at (6.2,0) [blackCirc] [label={below:$(2,3)$}] (g) {};
\node at (7.2, 2) [] (end1) {};
\node at (7.2, 0) [] (end2) {};
\draw [thick,->] (a.east) -- (b.west);
\draw [thick,->] (b.east) -- (c.west);
\draw [thick,->] (b.east) -- (g.west);
\draw [thick,->] (b1.east) -- (g.west);
\draw [thick,->] (e.east) -- (b1.west);
 \draw [thick,->]  (b) edge[out=260,in=110,->] (b1);
 \draw  [thick,->] (b1) edge[out=60,in=290,->] (b);
\end{tikzpicture}
\caption{$E'_{temp}$ for $\hat{S}$ of ~\figref{fig:barrier-run-state}}
\label{fig:etemp}
\end{subfigure}
\begin{subfigure}[b]{0.45\textwidth}
\centering
\begin{tikzpicture}[scale=0.65]
\node at (0,2) [blackCirc] [label={above:$(1,1)$}] (a) {};
\node at (3,2) [blackCirc] [label={above:$(1,2)$}] (b) {};
\node at (1.5,0) [blackCirc] [label={below:$(2,1)$}] (e) {};
\draw [thick,->] (a.east) -- (b.west);
\draw [thick,->] (b) -- (e);
\draw [thick,->] (e) -- (a);
\end{tikzpicture}
\caption{Event model generated from states of Fig.~\ref{fig:invalidB}}
\label{fig:eventcycle}
\end{subfigure}
\caption{$SE$ transform applied to earlier examples}
\end{figure}
\tikzstyle{place}=[circle,draw=black,fill=white,thick,inner sep=0pt,minimum size=2mm]
\begin{figure}[!tb]
\centering
\begin{tikzpicture}
\node at (0,0) [place] [label={left:$[i,s]$}] (is) {};
\node at (2,0) [] (is1) {};
\node at (0,1) [place] [label={left:$[j,t-1]$}] (jtmin1) {};
\node at (2,1) [place] [label={right:$[j,t]$}] (jt) {};
\node at (2,2) [place] [label={right:$[k,u]$}] (ku) {};
\node at (0,2) [] (kumin1) {};

\draw [thick,->] (is.east) -- (is1.west);
\draw [thick,->] (is.east) -- (jt.south);
\draw [thick,->] (jtmin1.east) -- (jt.west);
\draw [thick,->] (jtmin1.east) -- (ku.west);
\draw [thick,->] (kumin1.east) -- (ku.west);
\draw [thick,->] (ku.south) -- (is.north);
\end{tikzpicture}
\caption{Case $1$ when $\omega_3$ is violated (in proof of Theorem 3)}
\label{fig:case12}
\end{figure}
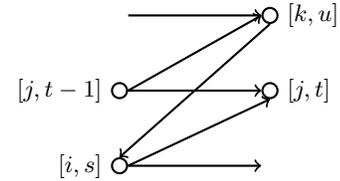
\tikzstyle{place}=[circle,draw=black,fill=white,thick,inner sep=0pt,minimum size=2mm]
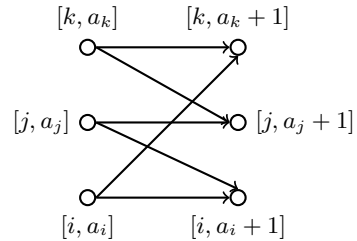
\begin{figure}
\centering
\begin{tikzpicture}
\node at (0,0) [place] [label={below:$[i,a_i]$}] (is) {};
\node at (2,0) [place] (is1) [label={below:$[i,a_i+1]$}] {};
\node at (0,1) [place] [label={left:$[j,a_j]$}] (jtmin1) {};
\node at (2,1) [place] [label={right:$[j,a_j+1]$}] (jt) {};
\node at (2,2) [place] [label={above:$[k,a_k+1]$}] (ku) {};
\node at (0,2) [place] [label={above:$[k,a_k]$}] (kak) {};
\draw [thick,->] (is.east) -- (is1.west);
\draw [thick,->] (is.east) -- (ku.south);
\draw [thick,->] (jtmin1.east) -- (is1.north);
\draw [thick,->] (jtmin1.east) -- (jt.west);
\draw [thick,->] (kak.east) -- (ku.west);
\draw [thick,->] (kak.east) -- (jt.west);
\end{tikzpicture} 
\caption{Illustration: Case 1 in proof of Lemma 5}
\label{fig:ijkal4}
\end{figure}

\balancecolumns
\end{document}